\documentclass[a4paper,UKenglish]{lipics-v2018}

\usepackage{mdframed} 
\usepackage{hyperref}
\hypersetup{colorlinks=true,linkcolor=blue,urlcolor=cyan,citecolor=green}
\usepackage{algorithmic}
\usepackage{amsmath, bm}
\usepackage{multirow}
\usepackage[ruled,vlined]{algorithm2e}
\newtheorem{thm}{Theorem}
\newtheorem{lem}[thm]{Lemma}
\newtheorem{prp}[thm]{Proposition}

\usepackage{float}
\usepackage{listings}

\newcommand{\OO}{\mathcal{O}\xspace}

\newcommand{\sO}{{\mathcal{O}^*}\xspace}

\newcommand{\volalg}{\textsc{volume}\xspace}
\newcommand{\annealing}{\textsc{annealing\_schedule}\xspace}
\newcommand{\billiard}{\textsc{billiard\_walk}\xspace}
\newcommand{\tests}{\textsc{perform\_tests}\xspace}
\newcommand{\estiratio}{\textsc{estimate\_ratio}\xspace}

\newcommand{\zapprox}{\textsc{compute\_z-approx}\xspace}

\DeclareSymbolFont{matha}{OML}{txmi}{m}{it}
\DeclareMathSymbol{\varv}{\mathord}{matha}{118}

\def\c++{{\tt C++}}

\def\vol{\mbox{vol}}
 
\def\cg{{\tt CoolingGaussian}}

\def\testR{{\tt L-test}}

\def\RR{{\mathbb R}}

\newcommand{\todo}[1]{}
\renewcommand{\todo}[1]{{\color{red} TODO: {#1}}}

\title{Practical volume estimation by a new annealing schedule for cooling convex bodies}
\titlerunning{Practical volume estimation}

\author{Apostolos Chalkis}
{Department of Informatics \& Telecommunications \\ National \& Kapodistrian University of Athens, Greece \\ GeomScale org}{achalkis@di.uoa.gr}{}{}
{}
\author{Ioannis Z.~Emiris}
{Department of Informatics \& Telecommunications\\
National \& Kapodistrian University of Athens, Greece}{emiris@di.uoa.gr}{}
{}{}
\author{Vissarion Fisikopoulos}
{Department of Informatics \& Telecommunications \\ National \& Kapodistrian University of Athens, Greece \\ GeomScale org}
{vfisikop@di.uoa.gr}
{}
{}{}

\authorrunning{A.~Chalkis et al.} 

\Copyright{Apostolos Chalkis, Ioannis Z.~Emiris and Vissarion Fisikopoulos}

\subjclass{Design and analysis of algorithms: \\Computational geometry, Random walks and Markov chains}


\acknowledgements{}

\usepackage{algorithmic}

\nolinenumbers
\hideLIPIcs
\begin{document}

\mdfsetup{linewidth=0.7pt}
\maketitle 

\begin{abstract}
We tackle the problem of efficiently approximating the volume of convex polytopes, when these are given in three different representations: H-polytopes, which have been studied extensively, V-polytopes, and zonotopes (Z-polytopes).
We design a novel practical Multiphase Monte Carlo algorithm that leverages random walks based on billiard trajectories, as well as a new empirical convergence tests and a simulated annealing schedule of adaptive convex bodies.
After tuning several parameters of our proposed method, we present a detailed experimental evaluation of our tuned algorithm using a rich  dataset containing Birkhoff polytopes and polytopes from structural biology.    
Our open-source implementation tackles problems that have been intractable so far, offering the first software to scale up in thousands of dimensions for H-polytopes and in the hundreds for V- and Z-polytopes on moderate hardware.
Last, we illustrate our software in evaluating Z-polytope approximations.
\end{abstract}

\section{Introduction}\label{sec:intro}

Volume computation is a fundamental problem with many applications~\cite{Cales18, Wong03,Iyengar88,Schellenberger09,venzke2019}. 
It is computationally hard: $\#$P-hard for explicit polytopes~\cite{Dyer88,volZono}, and APX-hard~\cite{Elekes1986} for convex bodies in the oracle model. 
Therefore, a great effort has been devoted to randomized approximation algorithms, based on Monte Carlo methods and Markov chains for sampling. The first celebrated result given in~\cite{DyerFrKa91} with complexity $\sO(d^{23})$, where $\sO(\cdot)$ suppresses
polylog factors and dependence on error parameters, and $d$ is the dimension.
Improved algorithms reduced the exponent to $5$~\cite{LovSim}.
The latter led to the first practical implementation~\cite{VolEsti} for high dimensional polytopes.
Further results~\cite{CousinsV14,LovVem} reduced the exponent to $3$, followed by another practical implementation~\cite{Cousins15}.
There are two limitations in current algorithm implementations. First, they only support polytopes given by a set of inequalities (H-polytopes) and, second, they scale typically only up to a few hundreds of dimensions.
Our aim is to provide the first practical method and its implementation that scales up to thousands dimensions for H-polytopes and to a few hundreds for other polytope representations.

Interestingly, Monte Carlo methods based on Markov chains have gain a lot of interest in the literature for several  important problems. For example, in maximum likelihood estimation~\cite{geyer_markov_1991}, machine learning \cite{brock2018large}, finance \cite{Cales18}, optimal control \cite{he2017numerical, huynh2012incremental}, Bayesian inference \cite{gamerman2006markov}, optimization~\cite{Martino16, schumer_adaptive_1968, Kalai06, Shcherbakov10, Bertsimas04}, and sampling~\cite{devroye_random_1984, martino_independent_2018, Lovasz06, neal2011mcmc, Narayanan20}.

Considering volume approximation, in the most general setting, a convex body is given by a (membership or separation) oracle. 
Thus, the complexity of a randomized volume algorithm is given by an upper bound on the number of oracle calls. 
Generally, a randomized algorithm uses a Multiphase Monte Carlo (MMC) technique that reduces volume approximation of a convex body $P$ to computing a product of ratios of integrals. 
The MMC we employ defines a sequence of convex bodies $P_k\subseteq\dots\subseteq P_0=P$. 
Then, the volume of $P$ is given by the telescoping product of volume ratios:
\begin{equation}\label{eq:telescopic_product}
\vol(P) = \vol(P_k) \frac{\vol(P_{k-1})}{\vol(P_k)} \cdots \frac{\vol(P_0)}{\vol(P_1)}.
\end{equation}
Clearly, the number of volume ratios is equal to the number of bodies in MMC.

Geometric random walks\footnote{i.e., Markov Chain Monte Carlo (MCMC) algorithms to sample from multivariate distributions constrained in (non-)convex bodies.} come into the picture to estimate each ratio by sampling from an appropriate distribution. 
A random walk in polytope $P$ starts at some interior point and, at each step, moves to a ``neighbouring" point, that we choose according to some distribution depending only on the current point. 
The complexity of a random walk is determined by two factors:
its mixing time, which equals the number of steps (or \emph{walk length}) required to make the distance between the current and the target distribution close to null, and the complexity of the basic operations that we perform at each step of the walk, called per step cost.

A crucial observation is that, if the ratios are bounded by some constant, 
then acceptance-rejection sampling would suffice to accurately estimate them, i.e.\ from Chebyshev inequality we obtain that for some \emph{volume ratio} $i<k$,
about $k\cdot(\vol(P_{i-1})/\vol(P_i))$ uniformly distributed points in $P_i$ suffice to estimate that volume ratio.
All in all, the total number of operations performed by an MMC volume algorithm is 
\begin{equation}\label{eq:total_complexity}
k \times\#\text{points per volume ratio}\, \times\#\text{operations per point}.
\end{equation}

The relations $rB_d\subseteq P\subseteq RB_d$, for $r<R$ and $B_d$ the $d$-dimensional unit ball, imply that $k=\lceil d\lg (R/r)\rceil$. 
We call $R/r$ the {\em sandwiching ratio} of $P$. 
For general convex bodies, the sandwiching ratio is minimized when the body is {\em well-rounded}, i.e.\ $B_d\subseteq P\subseteq \OO(\sqrt{d})B_d$. 
In~\cite{LovSim} they put $P$ in isotropic position, which implies well-roundness, and use scaled copies of the $d$-dimensional unit ball $B_d$ to define the MMC, i.e.\ $P_i=(2^{(k-i)/d}B_d)\cap P$, for $i=0,\dots ,k=\Theta(d\lg d)$.
These techniques combined with the so called speedy walk --a special case of the ball wall-- delivers an $\sO(d^5)$ algorithm.

In~\cite{LovVem} they fix a sequence of functions $f_0,\dots, f_k$ and $\vol(P)$ is given by a telescopic product of ratios of integrals generalizing Equation~(\ref{eq:telescopic_product}). 
They use a sequence of exponential functions; after putting $P$ in well-rounded position in $\sO(d^4)$ oracle calls, the number of integral ratios in MMC $k=\sO(\sqrt{d})$ and the total number of oracle calls is $\sO(d^4)$ by using Hit-and-Run (HR). 
In~\cite{CousinsV14} they employ a sequence of spherical Gaussian distributions and an improved mixing time for the ball walk. 
They assume that the input convex body is well-rounded and the total number of oracle calls is $\sO(d^3)$. 
In~\cite{Jia20} they provide an improved rounding algorithm which puts a convex body in well-rounded position in $O(d^{3.5})$ oracle calls (that becomes $O(d^3)$ after recent progress in KLS conjecture~\cite{chen20}). Thus, the latter complexity also determines the bound for volume approximation of general convex bodies. 
%
For H-polytopes, with $m$ facets, in~\cite{Lee18} they extend the simulated annealing of~\cite{LovVem} to the manifold setting and together with the Riemannian Hamiltonian Monte Carlo sampling, they obtain a $\OO^*(m^2d^{\omega-1/3})$ bound, where $\omega$ stands for the matrix multiplication exponent.
In~\cite{Mangoubi19}, they show a $\OO^*(md^{4.5}+md^{4})$ bound via a sub-linear ball walk and Gaussian cooling.

However, it is impractical\footnote{A recent exception for sampling is the implementation of~\cite{kook22}.} to use an implementation of the above volume approximation algorithms due to extremely large constants in the complexity and pessimistic upper bounds on the mixing time of the random walks they employ.
Considering practical volume algorithms, the objective is to perform efficient computations to obtain the volume of a high dimensional polytope. 
Typically, to achieve this goal, a practical algorithm on one hand, is based on a theoretical algorithmic scheme and on the other hand, provides practical adjustments, heuristics and efficient tuning. 
Moreover, practical algorithms instead of theoretical guarantees, provide strong empirical evidences for the run-time and the accuracy of the method through extended experimental results. 
For example, the method in~\cite{VolEsti} is based on the ball sequence in~\cite{DyerFrKa91} and performs $\sim md^4$ operations to achieve a relative error at most $0.1$ in practice. 
The method in~\cite{Cousins15} is based on the Gaussian cooling in~\cite{CousinsV14}, and performs $\sim md^3$ operations to achieve the same relative error.

Considering sampling in practice, the dominant paradigm for random walks is CDHR~\cite{kaufman98}. 
It is a version of HR that uses directions parallel to the axes. 
HR mixes after $\sO(d^3)$~\cite{Lovasz06} steps for isotropic log-concave distributions.
Very recently, the mixing time of CDHR has been bounded by $\sO(d^9(R/r)^2)$~\cite{Laddha20, Narayanan20}.
However, experiments~\cite{VolEsti, CousinsChnr} indicate that both CDHR and HR have similar convergence rates in practice. 
Thus, the faster step of CDHR compared to HR---$O(m)$ vs. $O(md)$---is the main reason why CDHR overshadowed, until recently, all other random walks in practical computations on H-polytopes.


A certainly interesting and overlooked problem is to compute in practice the volume of a polytope given by other representations, namely, V-polytopes (given by the set of extreme points or vertices) and Z-polytopes or zonotopes, given as a Minkowski sum of segments.
For those representations all known practical algorithms~\cite{VolEsti, Ge2015, Cousins15} fail to compute the volume for dimensions beyond, say, $d> 15$. 
There are two main reasons for that, besides the costly oracles (which are equivalent to a linear program). 
First, it is not possible to compute efficiently the largest inscribed ball needed as an initialization by all algorithms.
In particular, for a Z-polytope $P$ and a ball $B$, checking whether $B\subseteq P$ is in co-NP
while for $P$ being a V-polytope, given $p\in P$, the computation of the largest inscribed ball centered at $p$ is NP-hard~\cite{Vpolyinsc}. 
Second, the number of facets --which is requested by the implementation of~\cite{Cousins15}-- is typically exponential in $d$ for both Z- and V-polytopes. 

A powerful approach to obtain well roundness is to put $P$ in \emph{near isotropic position}~\cite{LovVem, CousinsV14}. 
In~\cite{adamczak10}, they prove that $\OO^*(d)$ uniformly distributed points in $P$ suffice to achieve $2$-isotropic position after some proper linear transformations. 
However, this could be quite computationally expensive in practice due to extensive uniform sampling from skinny polytopes.
In~\cite{CousinsChnr}, they provide an alternative method which achieves efficiency comparable in practice to that in~\cite{Cousins15}. 
They compute the maximum volume ellipsoid in $P$, they map it to the unit ball and apply to $P$ the same transformation. 
They experimentally show that a few iterations suffice to put $P$ in John's position~\cite{John48}. 
However, when $P$ is a V-polytope there is no known algorithm to compute the largest inscribed ellipsoid of $P$.
For a Z-polytope, the John ellipsoid $E$ is tighter, as for any centrally symmetric convex body compared to non-symmetric convex bodies. However, there is no algorithm for the computation of the John ellipsoid of a Z-polytope. 
The algorithm in~\cite{Cern12} computes an ellipsoid $E: \ E\subseteq P\subseteq dE$ for a given Z-polytope; the complexity is polynomial but the degree is not determined. 
Thus, it is unclear whether it may lead to an efficient rounding and whether it can be practical for Z-polytopes. 

Interestingly, for well-rounded convex polytopes, using any sequence of bodies to fix the telescopic product leads to $k = \Theta^*(d)$ phases in the telescopic product~\cite{LovVem}, while using a sequence of any log-concave functions leads to $k = \Theta^*(\sqrt{d})$ phases~\cite{CousinsV14}. 
Therefore, the fact that our practical algorithm is faster than that in~\cite{Cousins15} apparently contradicts theory~\cite{LovSim, CousinsV14}. The main reason behind this "contradiction" is the efficiency of the improved version of Billiard Walk we use. 
In particular, extended experiments have shown that ---for rounded convex polytopes--- Billiard Walk mixes after $O^*(1)$ steps in practice~\cite{chalkis2020geometric}, while its per step cost is $O(md)$. On the other hand, both HR and CDHR mix after $O^*(d^2)$ steps in practice~\cite{Cousins15, CousinsChnr}.

\subsection{Our contributions}\label{subsec:contributions}

We propose the first practical volume algorithm that scales to thousands of dimensions for H-polytopes outperforming all the existing practical algorithms. Additionally, it efficiently handles two more representations of polytopes, namely V- and Z-polytopes, up to a few hundred dimensions for the first time. 
We provide an extended experimental analysis that provides empirical evidences on the accuracy as well as the performance of our practical algorithm.
Moreover, our experiments show that our proposed practical algorithm outperforms the state-of-the-art, namely~\cite{Cousins15} (e.g.\ for unit cubes, 
20 times faster in $d=50$ and 95 times faster in $d=500$; for $d=1000$ our implementation is expected to be $\sim 130$ times faster; see section~\ref{subsec:h_poly_experiments} and Figure~\ref{fig:cb_cg_runtime_comparison} for more details).

The algorithm (Algorithm~\ref{alg:volume}) is based on the MMC framework and combines the sequence of convex bodies 
with a new simulated annealing schedule, while leveraging uniform sampling and the superior practical performance of Billiard walk~\cite{POLYAK20146123}.
In brief, the algorithm fixes a convex body $C\in\RR^d$ and a sequence $C_1\supseteq\cdots \supseteq C_k$ that intersects $P$, where $C_i = q_iC$, $q_i\in\RR_+$, 
using a cooling or annealing schedule (Algorithm~\ref{alg:annealing}).
The algorithm uses a telescopic product (Equation~(\ref{eq:teleprod2})) of ratios of volumes, each  bounded by a constant with high probability by the annealing schedule. 
With high probability, $k=\OO(\lg(\vol(P)/\vol(P_k)))$, where $P$ is the input polytope and $P_k = C_k\cap P$ the body with minimum volume in MMC (Section~\ref{subsubsec:number_of_phases}), which is not surprising, as---if $C$ is a ball---this results to $\sO(d)$ bodies in MMC with $\sO(d)$ points required per volume ratio.
However, our scheme offers crucial advantages in practice compared to the state-of-the-art by (a) optimizing the hidden constants in the complexity (Propositions~\ref{nballs1}, and~\ref{nballs2}), (b) taking advantage of fast practical uniform samplers, and (c) exploiting good choices of $C$ in the Z-polytope case.

To fix the sequence of $C_i$ we introduce a new simulated annealing schedule. The schedule minimizes the number of phases $k$ given a body $C$ with high probability ---assuming perfect uniform sampling. 
To achieve this we define a new statistical test that restricts each ratio in the telescopic product in Equation~(\ref{eq:telescopic_product}) in the interval $[r, r+\delta]$, where $r$ and $\delta$ are predefined constants.

For sampling, we leverage the Billiard walk for the first time in volume computation (Section~\ref{sec:sampling}). 
We show that, with the right selection of parameters, this walk outperforms CDHR even if the walk length is $1$. 
Using this walk length, our algorithm only needs $\sO(d)$ points per volume ratio (Section~\ref{sec:implementation} and Figures~\ref{fig:zono_compl}, and~\ref{fig:vpoly_compl}), to ensure the desired accuracy (Figure~\ref{fig:hnr_errors}). 
We provide efficient implementations for the step of \billiard\ for all polytope representations.
For Z- and V-polytopes, the steps are reduced to solving linear programs while for H-polytopes, with $m$ facets, we exploit the fast implementation of \billiard\ in~\cite{chalkis2020geometric} to extend it to the case of the intersection with a ball. Our efficient implementation requires $O(dm)$ operations per \billiard\ point.


\begin{figure}[t!]\centering
\includegraphics[width=0.45\textwidth]{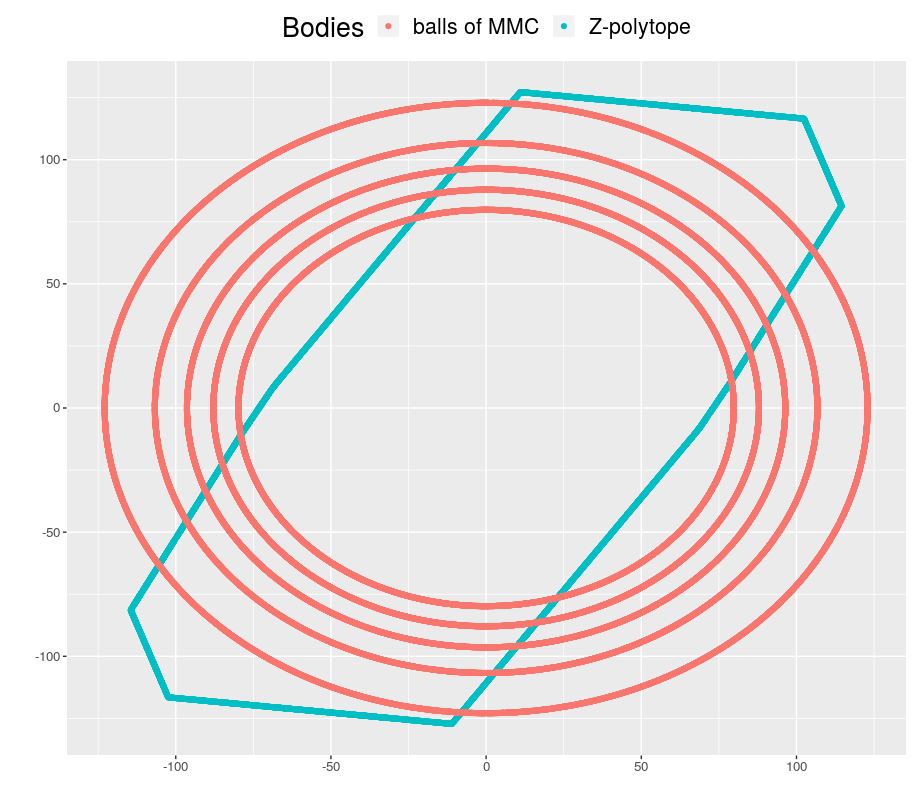}
\hspace{1cm}
\includegraphics[width=0.45\textwidth]{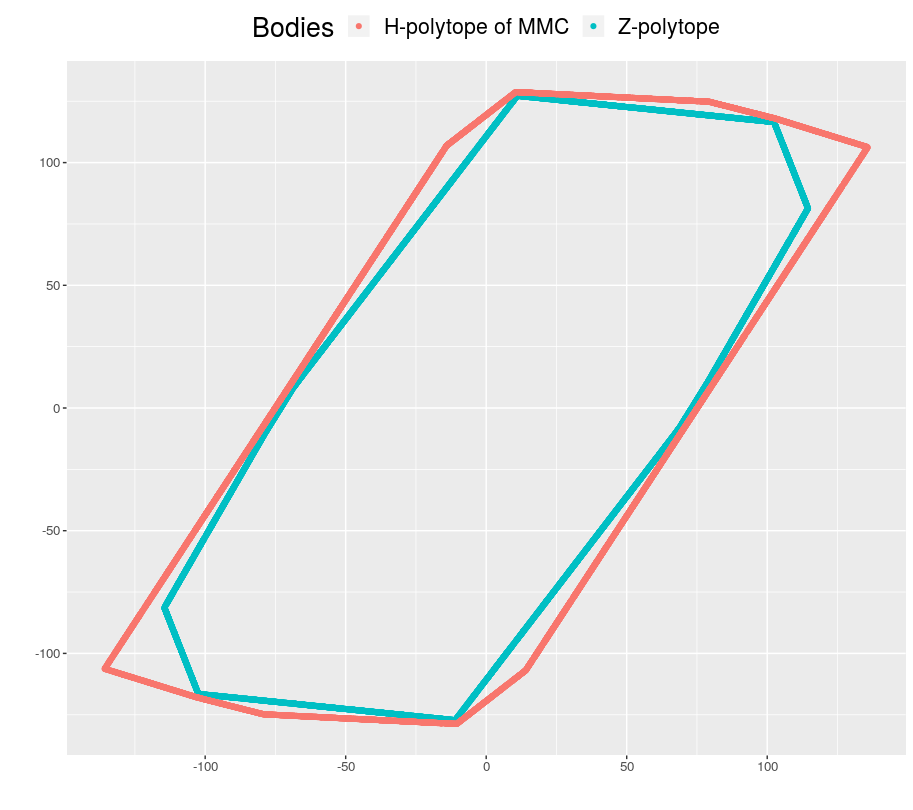}
\caption{An example of how our algorithm uses different bodies to fix the telescopic product (MMC) in Equation~(\ref{eq:telescopic_product}) and choses the most efficient option. When the input polytope is a Z-polytope, left, we use balls in MMC. Right, we use a centrally symmetric H-polytope that can be exploited to skip rounding and to reduce the number of bodies in MMC and total runtime. In both plots $\frac{\vol(C'\cap P)}{\vol(C')} \in [0.9,0.95]$, where $P$ is the Z-polytope, left $C'$ is the smallest ball, and right, $C'$ is the red centrally symmetric H-polytope.\label{fig:zono_ann}}
\end{figure}

For rounding, we propose methods for all types of polytopes (Section~\ref{sec:rounding}) that are faster than the ones currently available~\cite{VolEsti, Cousins15, CousinsChnr} as reported by our experiments (Section~\ref{sec:implementation}).
For H-polytopes, we simply enhance the method from~\cite{Cousins15} with the billiard sampler. 
For V-polytopes, our method brings $P$ to approximate John's position by computing the smallest enclosing ellipsoid using $P$'s vertices and the algorithm in~\cite{Todd07}.
For Z-polytopes, we construct an H-polytope that approximates $P$ well, and set it as $C$ in MMC.
We experimentally show that this choice, with no rounding, is more efficient than estimating the volume after applying various rounding methods.

We offer an optimized implementation of the proposed algorithm that scales in thousands of dimensions for H-polytopes and hundreds of dimensions for V- and Z-polytopes on moderate hardware (Section~\ref{sec:implementation}) for the first time.  
Our software performs computations which were intractable until now (Tables~\ref{fig:zono}, and~\ref{fig:vpoly}), outperforms current state-of-the-art implementations~\cite{VolEsti, Cousins15} in all representations, and is illustrated in evaluating Z-polytope approximations in the context of mechanical engineering~\cite{Kopetzki17}.
Our implementation\footnote{\url{https://github.com/GeomScale/volume_approximation/tree/v1.1.3}} builds upon and enhances {\tt volesti}, a \c++\ open source library for high-dimensional sampling and volume computation with an {\tt R} interface~\cite{Rvolesti}.

\paragraph{Notation} We denote the dimension that the polytope $P$ lies with $d$. For H-polytopes $m$ denotes the number of facets, For V- and Z-polytopes we denote the number of vertices and the number of generators respectively with $n$. Considering algorithms, with \volalg\ we refer to our practical algorithm and its implementation, and with \cg\ we refer to the implementation of the practical algorithm in~\cite{Cousins15}.

\section{Uniform sampling and polytope oracles}\label{sec:sampling}

In this section we introduce the proposed uniform sampler that the volume algorithm will use and discuss polytope operations.

We start with some notation. $P$ is a full-dimensional convex polytope lying in $d$-dimensional space and $\partial P$ is its boundary.
An H-polytope (in H-representation) is
\begin{equation}\label{eq:hpoly}
P=\{x\ |\ Ax\leq b,\ A\in\RR^{m\times d},\ b\in\RR^{m} \}
\end{equation} 
with $m$ facets in $d$ dimensions.
A V-polytope is the convex hull of a pointset in $\RR^d$; equivalently it can be seen as a linear projection, given by a matrix $V\in\RR^{d\times n}$, of the canonical simplex, i.e.\ 
\[\{x\in\RR^n\ |\ \sum_{i=1}^nx_i=1,\ x_i\geq 0 \},\] 
while the columns of $V$ corresponds to the $n$ vertices of $P$.
A zonotope (Z-polytope) is the Minkowski sum of $n$ $d$-dimensional segments or equivalently given by matrix $G\in\RR^{d\times n}$ and seen as a linear map of hypercube $[-1,1]^n$ to $\RR^d$: we call it a Z-representation. The {\em order} of a Z-polytope is the ratio $n/d$. We denote by $[n]$ the set $\{1,2,\dots,n\}$ for some natural number $n$. 

Random sampling is fundamental for efficient volume approximation. Our algorithm employs Billiard walk (Algorithm~\ref{alg:billiard}, \billiard) to sample approximately uniform points from a polytope $P$. \billiard\ starts from a given point $p_0 \in P$,
selects uniformly at random a
direction, say $v_0$, and it moves along the direction of $v_0$ for length $L$;
it reflects on the boundary if necessary. This results a new point $p_1$ inside
$P$. We repeat the procedure from $x_1$. Asymptotically it converges to the
uniform distribution over $P$. The length is $L=-\tau \ln \eta$, where $\eta$ is
a uniform number in $(0,1)$, that is $\eta\sim\mathcal{U}(0,1)$, and $\tau$ is a
predefined constant. It is useful to set a bound, say $\rho$, on the number of
reflections to avoid computationally hard cases where the trajectory may stuck
in corners. In \cite{POLYAK20146123} they set
$\tau \approx \mathop{diam}(P)$ and $\rho =10d$.

\begin{algorithm}[t]
  \caption{\billiard$(P, p, \rho, \tau, S)$}
	\label{alg:billiard}
	\SetKwInOut{Input}{Input}
	\SetKwInOut{Output}{Output}
    \SetKwInOut{Require}{Require}
    \SetKwRepeat{Do}{do}{while}

    \Input{polytope $P$; point $p$; upper bound on the number of reflections $\rho$; length of trajectory parameter $\tau$; walk length $S$.}

    \Require { point $p\in P$}

	\Output{A point in $P$}

    \BlankLine
    \For {$j=1,\dots ,S$} {
       $L \leftarrow -\tau\ln\eta$, $\eta\sim \mathcal{U}(0,1)$ \tcp{length of the trajectory}

       $i\leftarrow 0$ \tcp{current number of reflections}

       $p_0\leftarrow p$ \tcp{initial point of the step}
       
       pick a uniform vector $\varv_0$ from the boundary of the unit ball


    	\Do{$i\leq \rho$} {
    		 $\ell \leftarrow \{p_i + t\varv_i, 0\leq t\leq L\}$ \tcp{segment}
    		\lIf{$\partial P\cap\ell=\varnothing$} {
    			$p_{i+1} \leftarrow p_i+L\varv_i$ ; \textbf{break}
    		}
    		$p_{i+1} \leftarrow \partial P\cap\ell$  \tcp{point update}
    		the inner vector, $s$, of the tangent plane at $p$, s.t.\ $||s|| = 1$\;
    		$L \leftarrow L - |P\cap\ell|$\;
    		$\varv_{i+1} \leftarrow \varv_i - 2(\varv_i^Ts) s$ \tcp{direction update}
    		$i \leftarrow i+1$\;
    	}
    	\leIf{$i=\rho$} {
    		$p \leftarrow p_0$
    	}
    	{
    		$p \leftarrow p_i$
    	}
    }
    \KwRet $p$\;
\end{algorithm}

\billiard\ requires an access to a boundary oracle, that is, to compute the intersection of a ray $\ell:=\{p+t\varv,\ t\in\RR_+ \}$ with $\partial P$, where $p\in P$. Additionally, the walk needs the tangent plane on the intersection point to compute the reflection of the billiard trajectory.

When the input is an H-polytope, see Equation (\ref{eq:hpoly}), the inner vector of the boundary point's facet is given by a row of $A$.
To compute the boundary point $\partial P\cap\ell$ for the first reflection, we solve $a_j^T(p_0 + t\varv_0) = b_j,\ j\in[m]$, for $t$, and keep the smallest positive root, where $a_j$ is the $j$-th row of $A$.
This takes $dm$ arithmetic operations in total.
A straightforward approach for \billiard\ would consider that each reflection costs $\OO(md)$ and thus each step costs $\OO(\rho md)$.
However, we exploit the fast version of \billiard\ in~\cite{chalkis2020geometric}. It improves both {\it point} and {\it direction updates} (see related comments in Algorithm~\ref{alg:billiard}) by storing computations from the previous iteration combined with preprocessing.
The preprocessing computes the inner products of all possible pairs of normal vectors of the facets, that takes $m^2 d$ operations. Then, the first reflection of each \billiard\ step costs $\OO(md)$ operations and the rest ones cost $\OO(m)$ operations each. Thus, the amortized per-step complexity to sample from a H-polytope with \billiard\ becomes $\OO((\rho + d)m)$.

\begin{lem}~\cite{chalkis2020geometric} \label{lem:BW-step-cost}
  The amortized per-step complexity of the \billiard\ (Algorithm~\ref{alg:billiard}) for a polytope $P\in\RR^d$ with $m$ facets, given in H-representation, is $\OO((\rho + d)m)$ operations after preprocessing that takes $\OO(m^2d)$  operations, where $\rho$ is the maximum number of reflections per step.
\end{lem}
In the context of the volume algorithm (section~\ref{sec:method}) when the body used in MMC is a ball we have to sample from the intersection of a H-polytope with a ball using \billiard. 
Assuming that the center of the ball is the origin with radius $R$, to compute the intersection of the ray $\ell:=\{p+t\varv,\ t\in\RR_+ \}$ with the boundary of the ball we compute the positive root of the equation,
\begin{equation}
t^2 + 2(p^T\varv)t + \|p\|_2^2 - R^2 = 0 ,
\end{equation}
which takes $\OO(d)$ operations. 
The normal to the tangent plane of a ball at point $p$ is $p/ \|p\|_2$. Thus, to compute the reflection of the ray when it hits the boundary of the ball takes $\OO(d)$ operations. Since $m>d$ the amortized per-step complexity of \billiard\ is $\OO((\rho + d)m)$.

For both Z- and V-polytopes, given a ray $\ell$, computing $\partial P\cap\ell$ reduces to a linear program (LP). 
One may consider the pre-image of $\ell$, hence computation of $\partial P\cap\ell$ becomes an optimization problem in the region defined by the intersection of a linear subspace with the hypercube or the simplex respectively. Let $g_i$ be the generators of a Z-polytope, $i\in[n]$. The intersection of $\partial P\cap \ell$ is given by the following LP:
\begin{equation}\label{eq:z_poly_boc}
\begin{split}
\max\ t,\hspace{0.3cm}s.t.\ \hspace{0.2cm}  & p+t\varv = \sum_{i=1}^{n} x_ig_i,\\
    & -1\leq x_i\leq 1,\ i\in[n] .
    \end{split}
\end{equation}
If $v_i$ are the vertices of a V-polytope, $i\in[n]$, then the intersection of $\partial P\cap \ell$ is given as follows:
\begin{equation}\label{eq:v_poly_boc}
\begin{split}
\max\ t,\hspace{0.3cm}s.t.\ \hspace{0.2cm}  & p+t\varv = \sum_{i=1}^{n} x_iv_i ,\\
    & x_i\geq 0,\ \sum_{i=1}^nx_i=1,\ i\in[n] .
    \end{split}
\end{equation}

To compute the reflection of the ray we exploit the solutions of the above LPs. For V-polytopes we keep the vertices that correspond to the positive $x_i$, which generate the facet that $\ell$ hits. Then, we compute the normal to the hyperplane they define. Similarly, for Z-polytopes we keep the generators that correspond to $x_i\neq \{-1,1\}$. Computing the normal in both cases requires to solve a $d\times d$ linear system.

Volume estimation also requires a membership oracle which decides whether point $p$ lies in $P$ (see estimation of $r_{k+1}$ in Equation~(\ref{eq:teleprod2})). For an H-polytope one has to check at most $m$ inequalities, thus it costs $\OO(dm)$ operations. For both Z- and V-polytopes, membership oracles are feasibility problems reducing to LP. For a V-polytope, $p\in P$ if and only if the following region is feasible: 
\[\sum_{i=1}^n x_i v_i = p,\ x_i\geq 0,\ \sum_{i=1}^n x_i =1.\]
For a Z-polytope, $v_i$ corresponds to the generators of $P$ and the rest of the constraints become $-1\leq x_i\leq 1$.

\section{Preprocessing: to round or not to round}\label{sec:rounding}

We introduce two different preprocessing methods. First, a rounding method that transforms a polytope of any representation to a so-called near isotropic position. Second, we propose a preprocessing method for Z-polytopes without rounding.  

We employ the rounding in~\cite{Cousins15} to introduce a faster method that brings $P$ to near-isotropic position, using \billiard\ for uniform sampling instead of Hit-an-Run. The method samples uniformly from $P$, then computes a linear transformation that puts the sample to isotropic position and applies to $P$ the same transformation. To map the pointset to isotropic position, we compute the SVD decomposition of the matrix that contains the generated points in rows~\cite{Shiri20}. We iterate this procedure until the ratio of the maximum over the minimum singular value falls below a threshold, e.g.\ $4$.

Next, we propose a different preprocessing with no rounding that can be more efficient than rounding for some specific polytopes such as Z-polytopes (Section~\ref{sec:implementation}).
The key idea is the choice of $C$ at MMC so that it is both a good approximation of $P$ and we can efficiently sample from it. The first property affects the number of bodies in MMC while the second the runtime per volume ratio in the telescopic product in Equation~(\ref{eq:teleprod2}). 

When $P$ is a Z-polytope with generator matrix $G\in\RR^{d\times n}$, we compute a centrally symmetric H-polytope $C\subset P$ (Algorithm~\ref{alg:z_approx}). Since $P$ is the projection of the hypercube $[-1,1]^n$ onto $\RR^d$ by $G$, our method projects a linear subspace of $[-1,1]^n$ onto $\RR^d$ such that (a) it is easy to compute the H-representation of the underlying polytope $C$ and (b) we minimize the information loss for $P$ implied by the restriction of $[-1,1]^n$ to that subspace. We exploit the basic idea of Principal Component Analysis (PCA)~\cite{Jolliffe16} and obtain $C$ by projecting with $G$ the linear subspace spanned by the kernel of $G^TG$ restricted in $[-1,1]^n$. The matrix $G^TG$ has $n-d$ zero eigenvalues. We compute an orthonormal basis of its kernel matrix $E\in\RR^{n\times (n-d)}$, whose SVD is given below, and use the orthogonal complement $W_{\perp}$ defined as follows:
\[
E = USV^T = \left[
\begin{array}{c} W \\
W_{\perp} \\
\end{array} \right]^T \left[
\begin{array}{cc}
S_1 & 0 \\ 0 & 0 \\
\end{array} \right] V^T.\] 

\begin{algorithm}[t] \caption{\zapprox$(G)$}
	\label{alg:z_approx}
	\SetKwInOut{Input}{Input}
	\SetKwInOut{Output}{Output}
    \SetKwInOut{Require}{Require}
    \SetKwRepeat{Do}{do}{while}
	
    \Input{The generator matrix $G\in\RR^{d\times n}$ of Z-polytope $P$.}
	\Output{An H-polytope $C\supset P$.}
    
    \BlankLine
    Compute the eigenvectors of $G^TG$\;
    Let the eigenvectors of zero eigenvalues form $E \in\RR^{n\times (n-d)}$\;
    Compute an orthonormal basis for $E$, and the orthogonal complement $W_{\perp}$\;
    Let the  H-representation of $[-1,1]^n$ be
$\left\{ y \in \RR^n \,\Big |\, \begin{pmatrix} I_n \\ -I_n\end{pmatrix} y
  \leq \mathbf{1}_{2n\times 1} \right\}$\;
    $M \leftarrow \begin{pmatrix} I_n \\ -I_n\end{pmatrix}W^T_{\perp}(GW^T_{\perp})^{-1}\in\RR^{2n\times d}$\;
    $C \leftarrow \{ x\ |\ Mx\leq \mathbf{1}_{2n\times 1} \}$\;
    \KwRet $C$\;
\end{algorithm}

\noindent
In short, we use the linear map induced by $W_{\perp}$ to compute an H-representation of $C$. 
By construction, $C\subset P$ since the domain of $C$ is a subset of the domain of $P$ (subspace of the hypercube) under the same mapping $G$. 

For an illustration of an example of $C$ and $P$ see Figure~\ref{fig:zono_ann} (right).

\section{Volume algorithm} \label{sec:method} 

In this section, we detail the volume algorithm (Algorithm~\ref{alg:volume}) based on uniform sampling and the preprocessing of Sections~\ref{sec:sampling}, and~\ref{sec:rounding}.
The algorithm uses the following telescopic product:
\begin{equation}\label{eq:teleprod2}
\vol(P) = \frac{\frac{\vol(P_k)} {\vol(C_k)}} {\frac{\vol(P_1)}{\vol(P_0)}\frac{\vol(P_2)}{\vol(P_1)}\cdots \frac{\vol(P_{k})}{\vol(P_{k-1})}} \, \vol(C_k),
\end{equation}
where $P_0=P,\ P_i=C_i\cap P,\, i=1,\dots , k$.
The annealing schedule bounds each ratio $r_i$ by a constant, with high probability.
The algorithm uses a ratio estimation (Algorithm~\ref{alg:ratio_estimation}) to  estimate 
\[
r_i = \frac{\vol(P_i)}{\vol(P_{i-1})},\; 1\leq i\leq k,
\]
using \billiard\ (Algorithm~\ref{alg:billiard}) to sample uniformly from $P_{i}$, and acceptance/rejection in $P_{i+1}$. 
The algorithm also estimates the ratio
\[
r_{k+1} = \frac{\vol(P_k)} {\vol(C_k)} 
\]
using uniform sampling from $C_k$ and acceptance/rejection in $P_{k}$. 
Lastly, \estiratio\ uses a sliding window of the last $l$ ratios and an empirical criterion to declare convergence in a statistical sense. 

Algorithm \volalg\ is parameterized by: the error of approximation $\epsilon$, cooling parameters $r,\delta >0$, $0< r+\delta < 1$, significance level (s.l.) $\alpha >0$ of the statistical tests, the degrees of freedom $\nu$ for the t-student used in t-tests, and $N$ that controls the number of points $\nu N$ generated per $P_i$ all used in the annealing schedule.

\begin{algorithm}[t]
  \caption{\volalg$(P, \epsilon, r, \delta, \alpha, \nu, N, l)$}
	\label{alg:volume}
	\SetKwInOut{Input}{Input}
	\SetKwInOut{Output}{Output}
    \SetKwInOut{Require}{Require}
	
    \Input{A polytope $P$; requested error $\epsilon$; annealing parameters $r, \delta, \nu, N$; length of sliding window $l$.}

    
	\Output{An estimation of $\vol(P)$}
    
    \BlankLine
    	Let $q_{\min}C\subseteq P \subseteq q_{\max}C$\; 
	$\{ P_0,\dots ,P_k,C_k \} =$ \annealing$(P,C,r ,\delta ,\alpha ,\nu, N,q_{\min},q_{\max} )$\;
	Set $\epsilon_i,\ i=1,\dots ,k+1$ s.t.\ ${\sum_{i=0}^{k+1}\epsilon_i^2}=\epsilon^2$\;
	\For{$i=1,\dots k+1$} {
		\leIf {$i\leq k$} {
			$r_i=$ \estiratio($P_i,P_{i+1},\epsilon_i, k, l$)\;
		}
		{
			$r_{k+1}=$ \estiratio($C_k,P_k, \epsilon_{k+1} , k, l$)
		}
	}
	\KwRet $\vol(C_k)/r_1/\cdots /r_k\cdot r_{k+1}$\;

\end{algorithm}

\subsection{Annealing schedule for convex bodies} \label{subsec:annealing}

Given $P$ and $q_{\min}C \subseteq P \subseteq q_{\max}C$, \annealing\ generates the sequence of convex bodies $C_1\supseteq\cdots\supseteq C_k$ defining $P_i=C_i\cap P$ and $P_0=P$. The computation of $q_{\min},q_{\max}$ is not trivial but in Section~\ref{sec:implementation} we give practical choices depending on $C$ which are very efficient in practice. The main goal is to restrict each ratio $r_i$ to the interval $[r, r+\delta]$ with high probability.

We introduce some notions from statistics needed to define two tests and refer to \cite{Cramer46} for details.
Given $\nu$ observations from a r.v.\ $X\sim\mathcal{N}(\mu ,\sigma^2 )$ with unknown variance $\sigma^2$, the (one tailed) t-test checks the null hypothesis that the population mean exceeds a specified value $\mu_0$ using the statistic $t= \frac{\bar{x}-\mu_0}{s/\sqrt{\nu}}\sim t_{\nu -1}$, where $\bar{x}$ is the sample mean, $s$ the sample st.d.\ and $t_{\nu -1}$ is the t-student distribution with $\nu-1$ degrees of freedom. Given a significance level $\alpha>0$ we test the null hypothesis for the mean value of the population, $H_0:\mu\leq \mu_0$ against $H_1:\mu > \mu_0$. We reject $H_0$ if,
\begin{align*}
t\geq t_{\nu -1,\alpha}	\iff \bar{x}\geq \mu_0+t_{\nu-1,\alpha}{s}/{\sqrt{\nu}},
\end{align*}
which implies $\Pr [\text{reject }H_0\ |\ H_0\ \text{true} ] =\alpha$.
Otherwise we fail to reject $H_0$.
In the sequel, we define two statistical tests for a volume ratio, which can be reduced to t-tests:
\begin{align}
\text{U-test}(P_1,P_2)\quad & H_0: \vol(P_2)/\vol(P_1)\geq r+\delta \\ 
\text{L-test}(P_1,P_2)\quad & H_0: \vol(P_2)/\vol(P_1)\leq r
\end{align}

The U-test and L-test are successful if and only if null hypothesis $H_0$ is rejected, namely $r_i$ is upper bounded by $r+\delta$ or lower bounded by $r$, with high probability, respectively. 

If we sample $N$ uniform points from a body $P_{i-1}$ then r.v.\ $X$ that counts points in $P_i$, follows $X\sim b(N,r_i)$, the binomial distribution, and $Y=X/N\sim \mathcal{N}(r_i,r_i(1-r_i)/N)$ follows the Normal distribution. The Algorithm~\ref{alg:perform_tests} provides all the steps to perform both statistical tests.

\paragraph{Remark}\label{Rrule}
The normal approximation of r.v.\ $Y=X/N$ suffices when $N$ is large enough and we adopt the well known rule of thumb to use it only if $Nr_i(1-r_i)>10$. 

Then each sample proportion that counts successes in $P_i$ over $N$ is an unbiased estimator for $\mathbb{E}[Y]$, which is $r_i$. So if we sample $\nu N$ points from $P_{i-1}$ and split the sample into $\nu$ sublists of length $N$, the corresponding $\nu$ ratios are experimental values that follow $\mathcal{N}(r_i,r_i(1-r_i)/N)$ and can be used to check both null hypotheses in U-test and L-test. Let $\hat{\mu}$ be the sample mean of the ratios. If 
\begin{equation}
r+\delta -t_{\nu-1,\alpha}\frac{s}{\sqrt{\nu}}\geq\hat{\mu}\geq r +t_{\nu-1 ,\alpha }\frac{s}{\sqrt{\nu}}
\end{equation}
then both U-test and L-test are successful and $r_i$ is restricted to $[r,r+\delta]$ with high probability. 

\begin{algorithm}[t]
  \caption{\tests$(P_1, P_2, r, \delta, \alpha, \nu, N)$}
	\label{alg:perform_tests}
	\SetKwInOut{Input}{Input}
	\SetKwInOut{Output}{Output}
    \SetKwInOut{Require}{Require}
	
    \Input{bodies $P_1, P_2$; parameters $r, \delta, \alpha, \nu, N$.}

    
	\Output{results of L-test and U-test}
    
    \BlankLine
	Sample $\nu N$ uniform points from $P_1$\;
	Partition $\nu N$ points to lists $S_1 ,\dots ,S_{\nu}$, each of length $N$ \;
	Compute ratios $\hat{r}_i = |\{ q\in P_2 :\ q\in S_i\} | / \, N$, $i=1,\dots ,\nu$\;
	Compute the mean, $\hat{\mu}$, and st.d., $s$, of the $\nu$ ratios\;
	\leIf {$\hat{\mu}\geq r +t_{\nu-1 ,\alpha }\frac{s}{\sqrt{\nu}}$} { 	
		{\tt L-holds} $\leftarrow$ {\tt true}
	}
	{
		{\tt L-holds} $\leftarrow$ {\tt false}
	}
	\leIf {$\hat{\mu}\leq r +\delta - t_{\nu-1 ,\alpha }\frac{s}{\sqrt{\nu}}$} {
		{\tt U-holds} $\leftarrow$ {\tt true}
	}
	{
		{\tt U-holds} $\leftarrow$ {\tt false}
	}
	\KwRet ({\tt L-holds}, {\tt U-holds})\;

\end{algorithm}

Let us now describe \annealing\ in more detail: 
Given $q_{\min}C \subseteq P \subseteq q_{\max}C$, \annealing\ computes the sequence $C_1\supseteq \dots\supseteq C_k$. Each body $C_i$ is a scalar multiple of a given body $C$.  When $C$ is the unit ball, the body used in each step is determined by a radius.
The initialization step of \annealing\ computes the body with minimum volume, denoted by $C'$ which is set as $C_k$ at termination of \annealing. To be precise, the last body in the sequence shall be $C_k$, s.t.\ $r_{k+1} = \vol(P_k)/\vol(C_k) \in [r,r+\delta ]$ with high probability. 
In particular, at initialization, \annealing\ binary searches in the interval $[q_{\min},q_{\max}]$ to compute a scalar $q'$ s.t.\ $C'=q'C$ and both U-test($C',C'\cap P$) and L-test($C',C'\cap P$) are successful. 

To compute $C_i = q_iC$, \annealing\ computes the scalar $q_i$ using binary search in a proper interval, such that both U-test($P_{i-1},q_iC\cap P$) and L-test($P_{i-1},q_iC\cap P$) are successful, where $P_0=P,\ P_i = C_i\cap P$ and $i=1,\dots ,k$. As $C'=q'C$ is the body of smallest volume in the sequence of MMC the interval to binary search for $q_i$ has to be $[q',q_{i-1}],\ i=1,\dots, k$, where $q_0 = q_{\max}$. At termination the \annealing\ sets $q_k = q'$. Notice that the binary search in the interval $[q',q_{i-1}]$ implies that $\vol(C'\cap P) < \vol(P_i) < \vol(P_{i-1})$ and if both L-test and U-test are successful then $r_i\in [r,r+\delta]$ with high probability. 
Also recall that in each step of binary search which computes $q_i$, \annealing\ samples $\nu N$ points from $q_{i-1}C\cap P$ to check both  U-test and L-test.

\annealing\ employs $C'$ to decide stopping at the $i$-th step after computing $C_i$; if the criterion fails, the algorithm computes $C_{i+1}$.
In particular, it checks whether $\vol(P_i)/ \vol(C'\cap P) \geq r$ with high probability, using L-test.  
Formally,
\begin{center}
\textit{Stop in step $i$ if\, L-test($P_i,C'\cap P$) holds. Then, set $k=i+1$, $q_k=q'$ and $C_k = q_kC,\ P_k=C_k\cap P$}.\\
\end{center}
To perform L-test, $\nu N$ uniformly distributed points are sampled from $P_i$.

\subsection{Termination}

This section proves that \annealing\ terminates with constant probability.

\begin{thm}\label{thm:Terminate}
Let $J$ be the minimum number of steps by \annealing, corresponding to no errors occurring in the t-tests.
Let also \annealing\ actually perform $M\geq J$ steps and $\beta_{\max}$, $\beta_{\min}$ be the maximum and minimum among all $\beta$'s in the $M$ pairs of t-tests in the U-test and L-test, respectively.
Then, \annealing\ terminates with constant probability, namely:
\begin{align*}
\Pr[\text{\annealing\ terminates}]  
 \geq\ & 1 - 2\frac{\alpha(1-\beta_{\min})+\beta_{\max}}{1 - \alpha(1-\beta_{\min})+\beta_{\max}} - \\
& \frac{2\beta_{\max} -\beta_{\min}^2}{1 - 2\beta_{\max}-\beta_{\min}^2} .
\end{align*} 
\end{thm}

\begin{algorithm}[t]
  \caption{\annealing$(P, C, r, \delta, \alpha, \nu, N, q_{\min}, q_{\max})$}
	\label{alg:annealing}
	\SetKwInOut{Input}{Input}
	\SetKwInOut{Output}{Output}
    \SetKwInOut{Require}{Require}
    \SetKwRepeat{Do}{do}{while}
	
    \Input{bodies $P, C$; parameters $r, \delta, \alpha, \nu, N, q_{\min}, q_{\max}$.}

    
	\Output{A sequence of bodies in Multiphase Monte Carlo}
    
    \BlankLine
    $q_1\leftarrow q_{\min}$, $q_2\leftarrow q_{\max}$, $q\leftarrow (q_1+q_2)/2$\;
	\Do{{\tt true}} {
		//binary search for the initialization\;
		Binary search in $[q_{\min}, q_{\max}]$ for a scaling factor $q$ such that both U-test($qC$, $qC\cap P$) and L-test($qC$, $qC\cap P$) holds\;
		
	    $C'\leftarrow qC$\;
	}
	$P_0 \leftarrow P$, $i \leftarrow 0$, $q_{\min} \leftarrow q$\;
	\Do{{\tt true}} {
		//computation of the sequence\;
		\If{\testR($P_i$,$C'\cap P$) holds} {
			$k \leftarrow i+1$, $P_k \leftarrow C'\cap P$ and \textbf{break}\;
		}
		Binary search in $[q_{\min},q_{\max}]$ for a scaling factor $q$ such that both U-test($P_i$, $qC\cap P$) L-test($qC$, $qC\cap P$) holds\;
		$P_{i+1} \leftarrow qC\cap P$, $i \leftarrow i+1$\;
		$q_{\max} \leftarrow q$\;
	}
	\KwRet $\{P_0,\dots ,P_k,C'\}$\;
\end{algorithm}
\begin{proof}
We demonstrate halting of \annealing.
In the t-tests, errors of different types may occur, thus, binary search may enter intervals that do not contain ratios in $[r,r+\delta]$, which implies that there is a (smaller) probability that \annealing\ fails to terminate.
We bound the probability that \annealing\ enters inappropriate intervals in a certain step --and thus we bound the probability of failing to terminate-- by a constant. 
Let $\beta$ capture the power of a t-test: $pow = \Pr[\text{reject}\ H_0\ |\ H_0\ \text{false}] = 1-\beta$.

Let $M\geq J$. If \annealing\ fails to terminate after $M$ pairs of U-test and L-test, then some type I or type II error occurred in the t-tests. An error of type I occurs when the null Hypothesis is true and the test rejects it, while type II occurs when the null Hypothesis is false and the test fails to reject it.
The respective probabilities are $\Pr[\text{reject }H_0\ |\ H_0 \text{ true}]=\alpha$ and $\Pr[\text{fail to reject }H_0\ |\ H_0 \text{ false}]=\beta$, which is a value of the quantile function of t-student. For the latter probability we write $\beta_L,\ \beta_R$ for U-test and L-test, respectively. If, for a pair of tests, both null hypotheses are false then an error occurs with probability
\begin{align*}
p_1 &= \Pr[\text{error occurs}\ |\ \text{both }H_0\text{ false}] =\beta_L(1-\beta_R) + \beta_R(1-\beta_L) + \beta_L\beta_R \\ &= \beta_L + \beta_R - \beta_L\beta_R \leq 2\beta_{\max} - \beta^2_{\min}
\end{align*}
Similarly,
\begin{align*}
p_2 &= \Pr[\text{error occurs}\ |\ \text{U-test\ }H_0 \text{ false } \text{and L-test\ }H_0\text{ true}]\\ &= \alpha(1-\beta_L) + \alpha\beta_L + \beta_L(1-\alpha) \\ & = \alpha + \beta_L - \alpha\beta_L \leq \alpha + \beta_{\max} - \alpha\beta_{\min}
\end{align*}
\begin{align*}
p_3 &= \Pr[\text{error occurs}\ |\ \text{L-test\ }H_0 \text{ false } \text{and U-test\ }H_0\text{ true}]\\ &=  \alpha(1-\beta_R) + \alpha\beta_R + \beta_R(1-\alpha) \\ & = \alpha + \beta_R - \alpha\beta_R \leq \alpha + \beta_{\max} - \alpha\beta_{\min}
\end{align*}
Then, 
\begin{align*}
\Pr &[\text{\annealing\ fails to terminate}]
 \leq \sum_{i=1}^M p_1^i + \sum_{i=1}^M p_2^i + \sum_{i=1}^M p_3^i\\
&\leq \sum_{i=1}^{\infty} p_1^i + \sum_{i=1}^{\infty} p_2^i + \sum_{i=1}^{\infty} p_3^i 
 = \frac{1}{1-p_1} + \frac{1}{1-p_2} + \frac{1}{1-p_3} - 3 \\
&\leq 2\frac{\alpha(1-\beta_{\min})+\beta_{\max}}{1 - \alpha(1-\beta_{\min})+\beta_{\max}} + \frac{2\beta_{\max}-\beta_{\min}^2}{1 - 2\beta_{\max}-\beta_{\min}^2} .
\end{align*} \hfill $\square$
\end{proof}
\subsubsection{Number of bodies in MMC}\label{subsubsec:number_of_phases}

We give probabilistic bounds on the number of bodies in MMC. Let us assume that 
i) we sample perfect uniform points in each step of \annealing\ and ii) that \annealing\ terminates successfully. We offer a probabilistic upper bound on the number of bodies in MMC $k$ and then a probabilistic interval where $k$ lies.

\begin{prp}\label{nballs1}
Given a convex polytope $P\subset\RR^d$ with sandwiching ratio $R/r$, i.e.\ $rB_d\subseteq P \subseteq RB_d$, and cooling parameters $r,\delta$ such that $r+\delta <1/2$ and parameters $\alpha$, $N$, $\nu$, $q_{\min}$, $q_{\max}$ let $k$ be the number of convex bodies in MMC returned by Algorithm \annealing, when $C$ is the unit ball. Then $\Pr[k\leq \lceil d\lg (R/r)\rceil] \geq 2-\frac{1}{1-\gamma} = 1 -\frac{\gamma}{1-\gamma}$, where $\gamma = \alpha (1-\beta_{\min})$ and $\beta_{\min}$ is the minimum among all the values of the quantile function appearing in L-test.
\end{prp}
\begin{proof}
If $k > \lceil d\lg (R/r)\rceil$ holds then $k\geq 1$ type I errors of U-test occurred in \annealing\, i.e., $H_0$ holds but the test rejects it, while L-test was successful, i.e., $H_0$ is false and the test rejects it. Type I error occurs with probability $\alpha$ and the probability of the success of L-test is $1-\beta_R$. Let $\beta_{\min}$ be the minimum among the values of the quantile function appearing in all instances of L-test. Then,
\begin{align*}
 \Pr[k>\lceil d\lg (R/r)\rceil ] &\leq \sum_{i=1}^k (a(1-\beta_i))^i \leq \sum_{i=1}^{\infty} (a (1-\beta_{\min}))^i\\ &= \frac{1}{1-\alpha (1-\beta_{\min})}-1\Rightarrow \\ 
 \Pr[k\leq \lceil d\lg (R/r)\rceil ] &\geq 2-\frac{1}{1-\alpha (1-\beta_{\min})}\\ &= 1 - \frac{\gamma}{1-\gamma},
\mbox{ where } \gamma =\alpha(1-\beta_{\min}).
\end{align*} \hfill $\square$
\end{proof}

\begin{prp}\label{nballs2}
Let a convex polytope $P\subset\RR^d$, cooling parameters $r,\delta$ such that $r+\delta <1/2$, parameters $\alpha$, $N$, $\nu$, $q_{\min}$, $q_{\max}$, and $k$ be the number of bodies in MMC by \annealing\. Then, 
\[
\Pr\bigg[\bigg\lfloor\log_{\frac{1}{r+\delta}}\bigg(\frac{\vol(P)}{\vol(P_k)}\bigg)\bigg\rfloor\\ \leq k\leq \bigg\lceil\log_{\frac{1}{r}}\bigg(\frac{\vol(P)}{\vol(P_k)}\bigg)\bigg\rceil\bigg]\geq 1- \frac{\gamma_L}{1-\gamma_L} - \frac{\gamma_R}{1-\gamma_R},
\]
where $\gamma_L = \alpha (1-\beta^L_{\min})$, $\gamma_R = \alpha (1-\beta^R_{\min})$ and $\beta^R_{\min},\ \beta^L_{\min}$ are minimum among all values of the quantile function appearing in U-test and L-test, respectively.
\end{prp}

 \begin{proof} Let,
\begin{align*}
& \Pr\bigg[\bigg\lfloor\log_{\frac{1}{r+\delta}}\bigg(\frac{\vol(P)}{\vol(P_k)}\bigg)\bigg\rfloor\leq k \leq \bigg\lceil\log_{\frac{1}{r}}\bigg(\frac{\vol(P)}{\vol(P_k)}\bigg)\bigg\rceil\bigg] \\
&  = 1 - \Pr\bigg[k > \bigg\lceil\log_{\frac{1}{r}}\bigg(\frac{\vol(P)}{\vol(P_k)}\bigg)\bigg\rceil\ \text{or}\ k < \bigg\lfloor\log_{\frac{1}{r+\delta}}\bigg(\frac{\vol(P)}{\vol(P_k)}\bigg)\bigg\rfloor\bigg] \\ 
& = 1 - \Pr\bigg[k > \bigg\lceil\log_{\frac{1}{r}}\bigg(\frac{\vol(P)}{\vol(P_k)}\bigg)\bigg\rceil\bigg] - \Pr\bigg[k < \bigg\lfloor\log_{\frac{1}{r+\delta}}\bigg(\frac{\vol(P)}{\vol(P_k)}\bigg)\bigg\rfloor\bigg] .
\end{align*}

Similarly to the proof of Proposition~\ref{nballs1}, we have $\Pr\bigg[ k > \bigg\lceil\log_{\frac{1}{r}}\bigg(\frac{\vol(P)}{\vol(P_k)}\bigg)\bigg\rceil\bigg] \leq \frac{\gamma_R}{1 - \gamma_R}$, where $\gamma_R = \alpha(1-\beta_{\min}^R)$, and $\beta_{\min}^R$ is the minimum among all values of the quantile function appearing in L-test.

If $k<\bigg\lfloor\log_{\frac{1}{r+\delta}}\bigg(\frac{\vol(P)}{\vol(P_k)}\bigg)\bigg\rfloor$, then $k$ type-I errors of L-test occurred, where $k\geq 1$, while U-test was successful with probability $1-\beta$. This implies that $\Pr\bigg[ k < \bigg\lfloor\log_{\frac{1}{r+\delta}}\bigg(\frac{\vol(P)}{\vol(P_k)}\bigg)\bigg\rfloor\bigg] \leq \frac{\gamma_L}{1 - \gamma_L}$, where $\gamma_L = \alpha(1-\beta_{\min}^L)$, and $\beta_{\min}^L$ is the minimum among all values of the quantile function appearing in U-test. 

Putting everything together, it follows that:
\begin{align*}
\Pr\bigg[\bigg\lfloor\log_{\frac{1}{r+\delta}}\bigg(\frac{\vol(P)}{\vol(P_k)}\bigg)\bigg\rfloor\leq k \leq \bigg\lceil\log_{\frac{1}{r}}\bigg(\frac{\vol(P)}{\vol(P_k)}\bigg)\bigg\rceil\bigg]\geq 1 - \frac{\gamma_L}{1-\gamma_L} - \frac{\gamma_R}{1-\gamma_R}
\end{align*} \hfill $\square$
\end{proof}

Therefore, the number of bodies is $k=O(\lg(\vol(P)/\vol(P_k)))$ with high probability. This implies that $k$ decreases when $C$ is a good fit to $P$. Clearly, the body that minimizes the number of bodies is the one that maximizes $\vol(C_k\cap P)$.

\subsection{Ratio estimation}\label{subsec:ratio_esti}

\volalg\ requires that we estimate $k+1$ ratios, which have been formed by \annealing. 
This section describes how to perform these estimations (Algorithm~\ref{alg:ratio_estimation}). 
First, we bound the error in each ratio estimation in order to use it for the definition of the stopping criterion. From standard error propagation analysis, for each volume ratio $r_i = \vol(P_i) / \vol(P_{i-1})$ in Equation~\ref{eq:teleprod2}, we have to bound the corresponding error by $\epsilon_i$ such that 
\begin{equation}
{\sum_{i=1}^{k+1} \epsilon_i^2}=\epsilon^2 ,
\end{equation}
Then, the telescopic product in Equation~(\ref{eq:teleprod2}) approximates $\vol(P)$ with error $\le\epsilon$. In Section~\ref{sec:implementation}, we further discuss efficient error splitting. To estimate $r_i$ we generate $n$ \billiard\ points in $P_{i-1}$ and we compute the proportion of the points that lie in $P_i$. For the ratio $r_{k+1} = \vol(P_k)/\vol(C_k)$ we follow the same procedure, but we sample from $C_k$ and we count the number of point in $P_k$. Computationally, we would like to determine an as small as possible integer $n$ where we are within our target accuracy $\epsilon_i$.

The main issue that one should address, is that \volalg\ generates correlated samples/points in the polytope $P$ using \billiard. Therefore, the best known bounds
for the number of \billiard\ points required are far too large for practical computations. Thus, to estimate $r_i$, we use a sliding window to keep the last $l$ estimation values of the ratio. That is a queue with length $l$; each time a new sample point is generated by inserting the new ratio value of $\hat{r}_i$ and by popping out the oldest ratio value. 

To determine an empirical stopping criterion, first, we assume that we have $n$ i.i.d.\ uniformly distributed samples in $P_{i-1}$. Then, and we employ the binomial confidence interval for the estimator of $r_i$ to bound $n$. The number of points in $P_i$ follows the binomial distribution $b(n,r_i)$. A binomial proportion confidence interval for $\hat{r}_i$ ---the current approximation to $r_i$--- is given by $\hat{r}_i \pm z_{\alpha /2} \sqrt{\frac{\hat{r}_i(1-\hat{r}_i)}{n}}$,
where $z_{\alpha /2}$ is the $1-\alpha /2$ quantile of the Gaussian distribution. 
As $n$ increases, the interval tightens around $\hat{r}_i$, thus, stopping for that value of $n$ where 
$$
\frac{z_{\alpha /2} \sqrt{\hat{r}_i(1-\hat{r}_i)/n}}{\hat{r}_i - z_{\alpha /2} \sqrt{\hat{r}_i(1-\hat{r}_i)/n}} \leq \epsilon_i ,
$$
obtains an estimation of $r_i$ within error $\epsilon_i$ with probability $(1-\alpha)$, and  Equation~(\ref{eq:teleprod2}) would estimate $\vol(P)$ up to at most $\epsilon$ with probability $(1 - \alpha)^{k+1}$.


\begin{algorithm}[t]
  \caption{\estiratio$(P_1,P_2, \epsilon, k, l)$}
	\label{alg:ratio_estimation}
	\SetKwInOut{Input}{Input}
	\SetKwInOut{Output}{Output}
    \SetKwInOut{Require}{Require}
    \SetKwRepeat{Do}{do}{while}
	
    \Input{bodies $P_1, P_2$; requested error $\epsilon$; number of bodies in Multiphase Monte Carlo; length of the sliding window $l$.}

    
	\Output{An estimation of $\frac{\vol(P_2)}{\vol(P_1)}$}
    
    \BlankLine
    $p \leftarrow 1 - \sqrt[k+1]{3/4}$\;
    $convergence \leftarrow ${\tt false}\;
    $j \leftarrow 0$, $count\_in \leftarrow 0$\;
    set the queue $W$ to be the sliding window\;
    \Do{$convergence$ is {\tt false}} {
    	$j \leftarrow j+1$\;
    	Generate an approximate uniform point $q_j$ from $P_1$ using \billiard\;
    	\If{$q_j\in P_2$} {
    		$count\_in \leftarrow count\_in + 1$\;
    	}
    	$r_j \leftarrow \frac{count\_in}{j}$\;
    	push $r_j$ to $W$\;
    	\If{$W \text{ has } l \text{ elements}$} {
    		$s\leftarrow std(W)$\;
    		$a \leftarrow r_j-z_{p/2}s$\;
    	    $b \leftarrow r_j+z_{p/2}s$\;
    		\leIf{$(b-a)/a\leq \epsilon /2$} {
    			$convergence \leftarrow$ {\tt true}\;
    		}
    		{
    			pop $1_{st}$ element of $W$
    		}
    	}
    }
    \KwRet $r_j$\;
\end{algorithm}

Notice that $\sqrt{\hat{r}_i(1-\hat{r}_i)/n}$ is an estimator of the st.d.\ of all sample fractions of size $n$. Thus, to define an empirical criterion we replace that quantity with the st.d., $s$, of the last $l$ ratio values, i.e.\ the ratios stored in the sliding window. 
Moreover, we consider the average, $\bar{r}$, of the ratios in the sliding window. Finally, we stop sampling when $\bar{r}$ and the st.d.\ $s$ meet the criterion of Equation~(\ref{eq:convergence_criterion}): then we say they meet convergence. 
Clearly, for the first $l$ points sampled, we do not check for convergence. For a probability $p$ sufficiently close to $1$, the empirical criterion for declaring convergence is as follows:
\begin{align}\label{eq:convergence_criterion}
\frac{(\bar{r}+z_{p/2}s) - (\bar{r}-z_{p/2}s)}{\bar{r}-z_{p/2}s} = \frac{2z_{p/2}s}{\bar{r}-z_{p/2}s}\leq \frac{\epsilon_i}{2}.\;
\end{align}

In section~\ref{sec:implementation} we experimentally show that this empirical rule suffice to achieve error $\epsilon$ for $p=1-\sqrt[k+1]{3/4}$, and for $l = O(1)$, i.e.\ when the length of the sliding window is constant ---independent from the dimension.

\section{Implementation and experiments}\label{sec:implementation}

In this section we discuss our implementation. We perform extended experiments analysing various aspects such as the tuning of the algorithm's parameters and its practical complexity. Finally, we apply our software to compute the volume of high-dimensional Birkhoff polytopes, and we test the quality of approximation of various methods for low-order reduction of Z-polytopes.

We use the {\tt eigen} library~\cite{eigenweb} for linear algebra and {\tt lpsolve}~\cite{lpsolve} for solving linear programs.
All experiments were performed on a PC with {\tt Intel® Core™ i7-6700 3.40GHz 8~CPU} and {\tt 32GB RAM} running {\tt Ubuntu 18}.
Runtimes reported in the plots and tables are averaged over~10 runs unless otherwise stated.
We denote by \cg\ the implementation of \cite{Cousins15}.
Brief instructions for how to run the implementation and reproduce the computational results of this article are described in the Appendix (Section~\ref{apx:tutorial}).
All experiments are run without multi-threading. 

It is of special interest to create and maintain a database of convex polytopes to evaluate the performance and the accuracy of various algorithms for sampling and volume approximation. 
Our polytope database, presented in Table~\ref{tbl:polydatabase}, is constructed by merging and extending the databases used in~\cite{Cousins15, VolEsti, vinci}. 

For the special case of polytope that comes from a metabolic network, the polytope is given as a set of equalities and inequalities; that is a low-dimensional polytope. To obtain the full dimensional polytope $P\subset\RR^d$ we follow the preprocessing in~\cite{chalkis2020geometric} which we compute with library {\tt cobra}~\cite{cobra20}.

The estimation of the telescopic product in Equation~(\ref{eq:teleprod2}) might lead to numerical overflows and underflows in floating-point arithmetic. To avoid that our implementation estimates the following expression,
\begin{equation}\label{eq:log_sum}
\log(\vol(P)) = \log(\vol(C_k)) - \sum_{i=1}^k \log({r_i}) +\log(r_{k+1}) . 
\end{equation}

To evaluate the efficiency of our implementation, when the input polytope is a V- or a Z-polytope, we count the number of reflections (or boundary oracle calls) that \billiard\ performs to achieve a relative error $\leq \epsilon$. 
This is the number of Linear Programs (LP) in Equations~(\ref{eq:z_poly_boc},~\ref{eq:v_poly_boc}) that our implementation solves. 
When the input polytope is a H-polytope, reflections have different complexities i.e.\ the first takes $\OO(md)$ while the rest $\OO(m)$ operations. 
Therefore, we count the number of points that \billiard\ generates to achieve a relative error $\leq \epsilon$.

\begin{table}[t]
\begin{tabular}{p{3cm}p{5.5cm}cc}
Polytope & Definition & H-rep. & V-rep.\\\hline\hline
cube-$d$ &  $\{x=(x_1,\dots,x_d)\, |\, -1\leq x_i\leq1, x_i\in\RR,\, i=1,\dots,d \}$ & $\checkmark$ & $\checkmark$\\\hline
cross-$d$ & cross polytope, the dual of cube, i.e.\
$\mbox{conv}(\{-e_i,e_i : \, i=1,\dots,d\})$ & $\checkmark$ & $\checkmark$\\\hline
$\Delta$-$d$ & $d$-dimensional simplex $\mbox{conv}(\{e_i : \, i=1,\dots,d\})$  & $\checkmark$ & $\checkmark$\\\hline
$\Delta$-$d$-$d$ & the product of two simplices, i.e.\ $\{ (p,p')\in\RR^{2d}\ |\ p\in\Delta -d,\ p'\in\Delta -d \}$,  & $\checkmark$ & \\\hline
 $\mathcal{B}_n$ Birkhoff polytope & $\{ x\in\RR^d\ |\ x_{ij}\geq 0,\ \sum_i x_{ij}=1,\ \sum_j x_{ij}=1,\ 1\leq i,j\leq n \}$  & $\checkmark$ & \\\hline
 $E$-$n$-$s$ Everest polytope & see \cite{Kerber17}  &  & $\checkmark$\\\hline
 $K$-$d$ dual Knapsack polytope & $\mbox{conv}(\pm e_1, \dots , \pm e_d, \mathbf{\alpha})$, where $\mathbf{\alpha}\in\mathbb{Z}^d_{\geq 0}$ is generated randomly, &  & $\checkmark$\\\hline
 $cc$-$8$-$k$ & 8d product of two 4d cyclic polyhedra   &  & $\checkmark$\\\hline
 $ccp$-$k$ & complete cut polytope on $k$ vertices  & & $\checkmark$\\\hline
 rvc-$d$-$n$ & polytope with $n$ vertices in $[-1,1]^d$: generate uniform points in $[-1,1]^d$, stop when $n$ lie in convex position  &  & $\checkmark$\\\hline
 rvs-$d$-$n$ & polytope with $n$ vertices uniformly distributed on the unit sphere   & $\checkmark$\\\hline
 rhs-$d$-$n$ & polytope with $n$ facets, with normals uniformly distributed on the unit sphere  & $\checkmark$ & \\\hline
 metabolic polytopes & polytopes that correspond to the flux space of a metabolic network~\cite{chalkis2020geometric,CousinsChnr}, see Table~\ref{tab:h_poly_vol} & $\checkmark$ &\\\hline
$Z_{\mathcal{U}}$-$d$-$n$ & $n$ generators, length of each generator selected uniformly from $[0,100]$ && \\\hline
$Z_{\mathcal{N}}$-$d$-$n$ & $n$ generators, length of each generator from $\mathcal{N}(50,(50/3)^2)$ truncated to $[0,100]$ && \\\hline
$Z_{Exp}$-$d$-$n$ & $n$ generators, length of each generator from $Exp(1/30)$ truncated to $[0,100]$ && \\\hline
\end{tabular}
\caption{The polytope database, where $d$ is the dimension,  
$e_1, \dots , e_d$ are the standard basis vectors in $\RR^d$.
The direction of each generator for Z-polytopes: $Z_{\mathcal{U}}$-$d$-$n$, $Z_{\mathcal{N}}$-$d$-$n$, $Z_{Exp}$-$d$-$n$, is chosen uniformly from the $(d-1)$-dimensional unit hypersphere.
\label{tbl:polydatabase}}
\end{table}

\subsection{Parameter tuning}\label{subsec:tuning}

Let us explain how we fine-tune the volume algorithm and the sampling procedures presented in Sections~\ref{sec:sampling}, and~\ref{sec:method}.

To start sampling from each body $P_i=C_i\cap P$ in MMC of Equation~(\ref{eq:teleprod2}) we use a central point of $P$. For H-polytopes we use the center of the Chebychev ball, i.e.\ the largest inscribed ball in $P$, which requires to solve a linear program \cite{Boyd04}. For V-polytopes, we compute an approximation of the minimum-volume enclosing ellipsoid of the vertices \cite{Todd07} and use its center. For Z-polytopes we use the center of symmetry as they are centrally symmetric convex bodies.

\subsubsection{Billiard walk's parameters}\label{subsubsec:sampling}

For \billiard\ first we have to set the parameter $\tau$ that controls the length of the linear trajectory of each step of the random walk. 
In~\cite{POLYAK20146123} they suggest to set $\tau$ equal to the diameter of the body we are about to sample. 
Since computing the diameter could be hard, we are suggesting the following heuristics depending on the representation of the input convex body.

\begin{enumerate}
\item For the intersection of a polytope $P$, given in any representation, with a ball $B$, we set $\tau$ equal to the diameter of $B$. 
\item
For H-polytopes computing the diameter is hard even by randomized algorithms~\cite{Brieden98}. Thus, we set $\tau=4*\sqrt{d}r$, where $r$ is the radius of the largest inscribed (Chebychev) ball in $P$.
\item
For V-polytopes, the computation of the diameter is straightforward and takes $O(dn^2)$ operations, where $n$ is the number of vertices. 
\item For Z-polytopes the diameter can be computed by a non-convex optimization problem:
\begin{equation}\label{eq:nonconv}
\min\ -\|x\|_2^2 = -\|Gy\|_2^2 = -y^TG^TGy,
\text{ subject to: } y\in [-1,1]^n .
\end{equation}
This optimization is NP-hard. However, we use the following heuristic: a) compute the covariance matrix of the generators of $P$, b) compute its eigenvector $w$ that corresponds to the maximum eigenvalue, c) compute $\max \langle w,x \rangle,\ x\in P$, in $O(dn)$, where $n$ is the number of generators. 
\item To compute the diameter of the intersection of a Z-polytope with a scaled copy of the polytope computed by Algorithm~\ref{alg:z_approx}, i.e.\ $Z-approx: =\{x\in\RR^d\ |\ Mx\leq b \}$, the program in Equation~(\ref{eq:nonconv}) becomes:
\begin{equation}\label{eq:nonconv2}
\begin{split}
\min\ -\|x\|_2^2 = -\|Gy\|_2^2 = -y^TG^TGy,
\text{ subject to: } y \in [-1,1]^d\cap \Phi,\\ \text{ where }\Phi :=\{y\in\RR^d\ |\ MTy\leq  b\} .
\end{split}
\end{equation}
In practice, we sample a set $S$ of $O(d^2)$ points from $[-1,1]^d\cap \Phi$, with CDHR, keeping the boundary points, and return $2\max\limits_{x\in S} ||x||^2_2$.
\end{enumerate}
Both heuristics for Z-polytopes and $P\cap Z-approx$ return better solutions in our experiments than solvers like {\tt Sequential QP}~\cite{NLoptSQP} from package {\tt NLopt} for the optimization problems of Equations~(\ref{eq:nonconv}) and~(\ref{eq:nonconv2}).

The second parameter of \billiard\ is the upper bound on the number of reflections $\rho$. In our experiments, the average number of reflections in a \billiard's step increases linearly or sub-linearly with the dimension (see Figures~\ref{fig:zono_compl},~\ref{fig:steps_comparison}, and~\ref{fig:vpoly_compl}). Thus, we tested upper bounds given by a linear function $\sim \gamma d$ for some $\gamma\in\mathbb{N}$. 
We choose to set $\gamma = 20$ by following the experimental evidences in Figures~\ref{fig:vpoly_compl},~\ref{fig:zono_compl} where $\gamma < 16$ up to dimension $100$.

Last, we have to set the walk length, which is the number of \billiard's steps we burn until we use a point for ratio estimation. 
We aggressively set this to $1$ following the observation in~\cite{Cousins15} that the fastest convergence of the empirical distribution happens for this particular value of walk length. 

Interestingly, in our experiments, the number of points that \billiard\ with walk length equals to $1$ generates to estimate each volume ratio, depends only on the number of bodies $k$ in MMC and not on the dimension $d$. 
For example, in Figure~\ref{fig:zono_compl} and Figure~\ref{fig:vpoly_compl} (down left plots), when \volalg\ fixes the same number of bodies in MMC, for different dimensions, \billiard\ generates the same number of points to achieve a relative error $\leq 0.1$. Consequently, the mixing rate of \billiard\ does not depend on the dimension on those instances.


\subsubsection{Annealing schedule's parameters}

For the cooling parameters in Section~\ref{subsec:annealing}, we set $r=0.1$ and $\delta=0.05$ and thus each volume ratio would be restricted in $[r,r+\delta]=[0.1,0.15]$ with high probability. 
We choose the significance level to be $\alpha =0.10$. A smaller $\alpha$ can be chosen for a tighter test around $r+\delta$ which will result to more iterations in each step of annealing schedule.
For larger $\alpha$ values the number of iterations could be reduced but the method becomes unstable, as the probability that \annealing\ fails to terminate increases while $\alpha$ increases as well (see Theorem~\ref{thm:Terminate}). 

To set the parameters $\nu, N$ we follow the empirical rule of Remark~\ref{Rrule}. Thus, assuming perfect uniform sampling and $r_i=0.1$, one has to set $N\geq 112$. At the initialization step of the annealing schedule we have to sample from a ball we set $N=120$ as we employ perfect uniform sampling from a ball. Otherwise, when \billiard\ is used in all the other cases of uniform sampling from a convex body we set $N=125$. We also set $\nu=10$ to keep the product $\nu N$ as small as possible. Smaller $\nu$ values would result to more unstable iterations in \annealing. Hence, following well known empirical rules we use t-student distribution in the t-tests.  
We apply the following optimization, for the $i$-th volume ratio, we sample from $P_i$ and if the stopping criterion fails, we employ binary search by reusing the same sample. Hence we sample only once the $\nu N$ points per step of \annealing. 

\annealing\ also has to compute $q_{\min},q_{\max}$ such that $q_{\min}C \subseteq P\subseteq q_{\max}C$. Ideally, the first would correspond to the largest inscribed scaled copy of the body $C$ in the polytope $P$. The second to the minimally scaled copy of $C$ that encloses $P$. In our implementation we set $q_{\min}=0$ and to set $q_{\max}$ we sample $N\nu$ points and compute $q_{\max}$ such that all points belong to $q_{\max}C$. In the special case of $C$ being the H-polytope of Algorithm~\ref{alg:z_approx} we compute the scaling factor $q_{\max}$ to compute the polytope $q_{\max}C$ that tightly encloses the Z-polytope $P$ and we set $q_{\min}=0$.

\subsubsection{Error splitting and sliding window}\label{subsubsec:error_splitting}

\begin{figure}[t]
\centering
\includegraphics[width=0.6\textwidth]{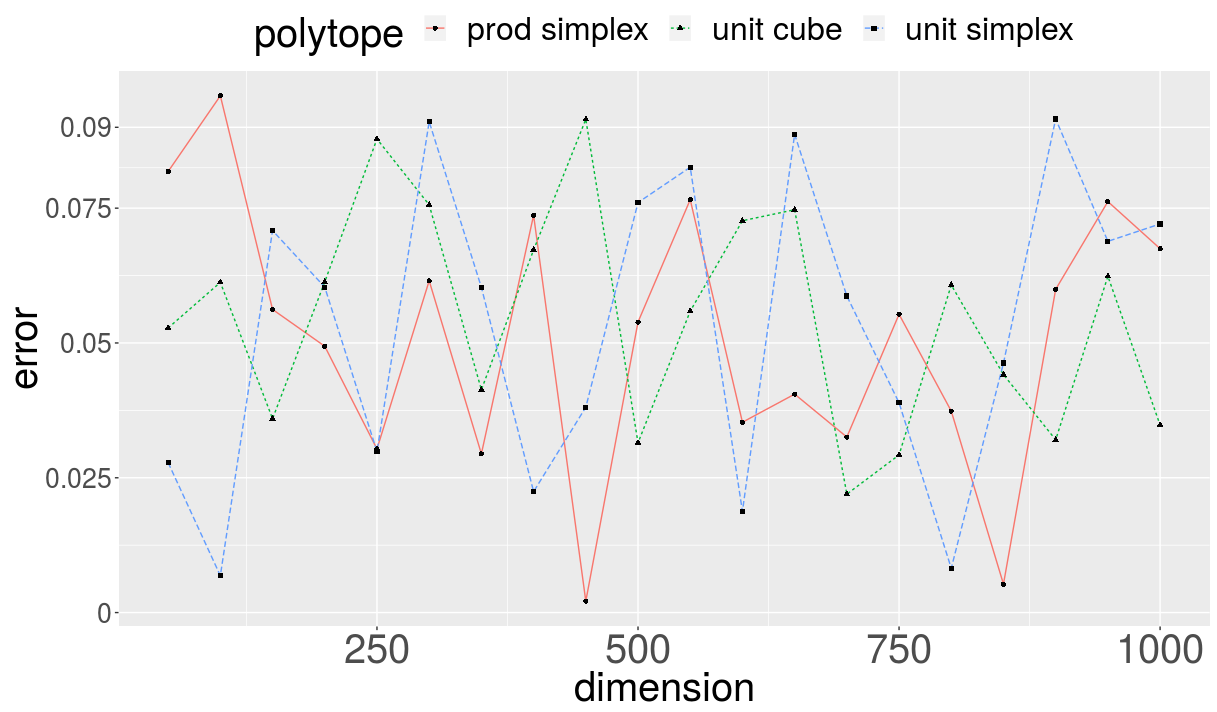}
\caption{The average errors of \volalg\ for cube-$d$, $\Delta$-$d$ and $\Delta$-$d$-$d$. This figure shows that the algorithm does not exceeds the input error parameter $\epsilon =0.1$ in practice.\label{fig:hnr_errors}}
\end{figure}

Recall that when we estimate each ratio in Multiphase Monte Carlo (MMC), we have to do an error splitting for the volume ratios that appear in the telescopic product in Equation~(\ref{eq:teleprod2})---i.e.\ given a requested error $\epsilon$ for the algorithm compute an error $\epsilon_i$ for each volume ratio $r_i$ (see Section~\ref{subsec:ratio_esti}), such that,
$$\sum_{i=1}^{k+1}\epsilon_i^2 = \epsilon^2.$$

We do not split the error $\epsilon$ equally to all ratios. In particular, we set the requested error $\epsilon_{k+1}=\epsilon/2\sqrt{k+1}$ for the last ratio $r_{k+1} = \vol(P_k)/\vol(C_k)$ to be the smallest ratio in the telescopic product, where $k$ is the number of bodies in MMC and the rest errors $\epsilon_i,\ i=1,\dots, k$ are set to be equally weighted. 
The reason is that $\vol(P_k)/\vol(C_k)$ converges faster than the other ratios in practice. The latter occurs because sampling from $C_k$ is usually faster and more accurate than sampling from any other $P_i$, 
e.g.\ when the body $C$ used in MMC is either a ball while $P_i$ is the intersection of $P$ with a ball, or an H-polytope while $P$ is a Z-polytope. 

Then we have to split $\epsilon '=\epsilon\sqrt{4(k+1)-1}/2\sqrt{k+1}$ to the remaining ratios, since $\epsilon_{k+1}^2 + \epsilon '^2 = \epsilon^2$. Thus we set $\epsilon_i=\epsilon '/\sqrt{k},\ i=1,\dots ,k$ so that ${\sum_{i=0}^{k+1} \epsilon_i^2}=\epsilon^2$ holds.

In the special case of $P$ being a Z-polytope, the body $C_k$ is an H-polytope. We estimate $\vol(C_k)$ by calling \volalg\ using balls in MMC, while the computational time of $\vol(C_k)$ is a small portion of the total runtime. For the computation of $\vol(C_k)$ we set the requested error $e''=\epsilon/2\sqrt{k+1}$ and then we equally split $\epsilon ' =\epsilon \sqrt{2k+1}/\sqrt{2k+2}$ to the $k+1$ ratios respecting ${\sum_{i=0}^{k+1} \epsilon_i^2}=\epsilon^2$.

For the length of the sliding window, we set $l=250$. Our experiments on the error show this choice offers stability, e.g.\ in Figure~\ref{fig:hnr_errors}   \volalg\ always compute a smaller error than the requested error $0.1$ for unit cubes, unit simplices and cross polytopes never exceeds the requested value for $d\leq 1000$.

For each new generated point, we update the average mean and the variance of the sliding window in $O(1)$  as follows. Let $\hat{\mu}$ be the average mean of the ratios in the sliding window, then we write the variance as \[\frac{1}{l}\sum_{i=1}^{l}(\hat{r}_i-\hat{\mu})^2=\frac{1}{l}(\sum_{i=1}^{l}\hat{r}_i^2-2\hat{\mu}\sum_{i=1}^{l}\hat{r}_i+l\hat{\mu}^2).\]
We store the sum of the window's ratios $\sum_{i=1}^{l}\hat{r}_i$ and the sum of the squared ratios $\sum_{i=1}^{l}\hat{r}_i^2$. For each new generated point we obtain an updated ratio and the oldest ratio is popped out. We use both the updated and the popped out ratios to update both $\sum_{i=1}^{l}\hat{r}_i^2$ and $\sum_{i=1}^{l}\hat{r}_i$ and to compute the updated average mean value and st.d. of the current ratios in the sliding window.   

In the sequel, we benchmark the proposed algorithm \volalg\ tuned as described in the previous sections. 
We present the result obtained from our experiments for each polytope representation separately.

\begin{figure}[!t]
\begin{minipage}[h]{0.49\textwidth}
\includegraphics[width=\linewidth]{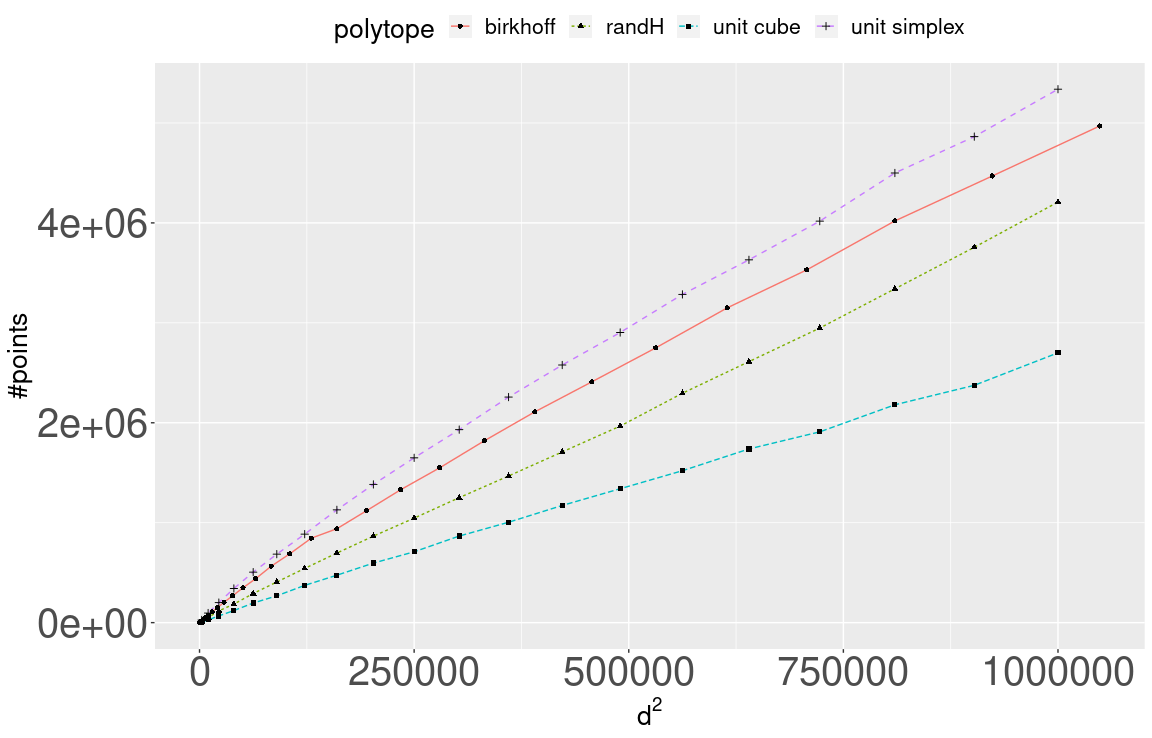}
\end{minipage}
\begin{minipage}[h]{0.49\textwidth}
\includegraphics[width=\linewidth]{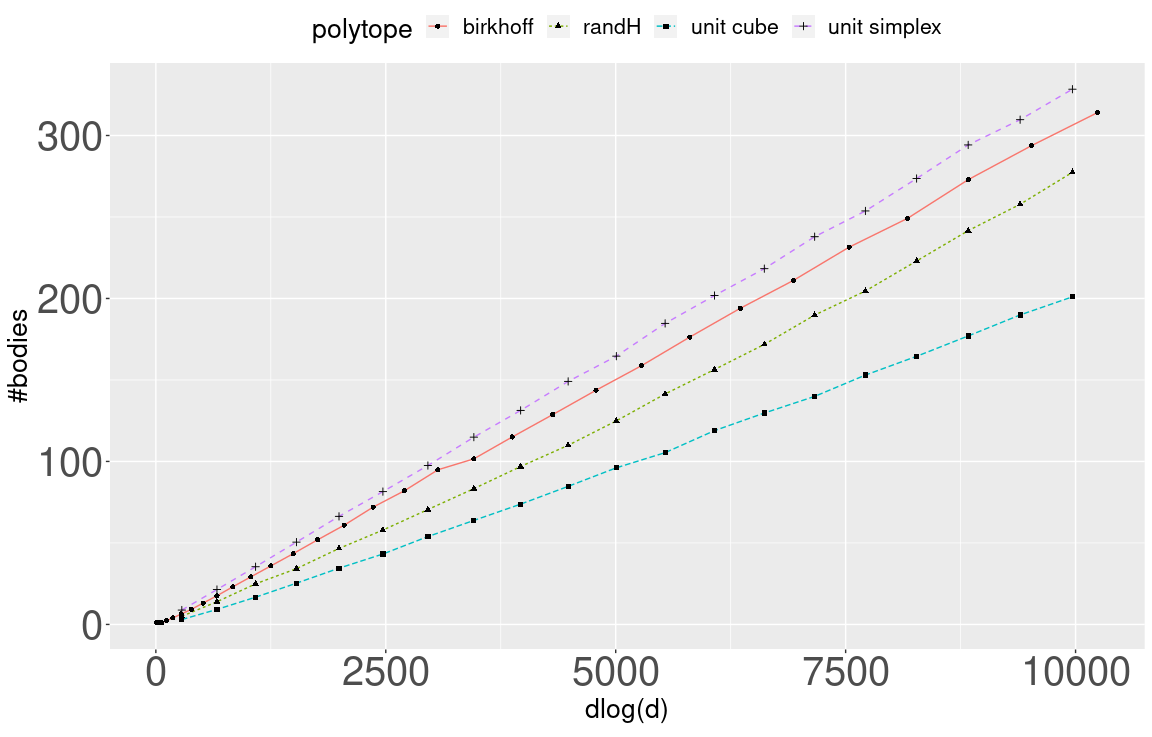}
\end{minipage}
\caption{Experimental complexity of \volalg\ for H-polytopes: the number of points that \billiard\ generates (left) and the number of bodies computed in MMC (right). We consider the cases of cube-$d$, $\Delta$-$d$ and rhs-$d$-$m$ with $m=5d$ and $d=50,100,\dots ,1000$. We set the error parameter $\epsilon = 0.1$ as input for the computations. \label{fig:steps_cb_cube_simplex_randh}}
\end{figure}

\subsection{Volume experiments for H-polytopes}\label{subsec:h_poly_experiments}

The implementation of our algorithm, \volalg, scales up to thousands of dimensions within a few hours. In Tables~\ref{tab:h_poly_vol},~\ref{tab:birkhoff_vol} we report the runtime of \volalg\ for cube-$d$, $\Delta$-$d$, $\Delta$-$d$-$d$, $rhs$-$m$-$d$, metabolic polytopes and $\mathcal{B}_n$, while \volalg\ achieves relative error equal at most $0.1$. Except of cube-$d$, $\Delta$-$d$, $\Delta$-$d$-$d$ we apply the rounding preprocess on Section~\ref{sec:rounding} before volume estimation. Our implementation takes for $d = 100$ a few seconds and for $d = 500$ a few minutes except for {\tt iSDY\_1059} that takes almost an hour as it has a larger number of facets (i.e.\ $m\approx 3000$). 
When $d=1000$ our implementation takes a few hours, while the runtime increases with the number of facets. However, the runtime for cube-$1000$ is smaller than the $\Delta$-$1000$ as the unit simplex has a larger isotropic constant than the unit hypercube; consequently the number of bodies in MMC is larger for the unit simplex. 
Our implementation is the first one that estimates the volume of very high dimensional Birkhoff polytopes extending the computational results of~\cite{VolEsti,Cousins15}.
In particular, the $\mathcal{B}_n$ of order $16\leq n\leq 33$ with dimension $225\leq d\leq 1024$ are computed in $\leq 4$hr.

To evaluate the performance of our implementation of \volalg\ we count the number of points that \billiard\ generates for the cube-d, $\Delta$-d, rhs-$d$-$m$ and $\mathcal{B}_n$. First, for the number of bodies in MMC we notice in Figure~\ref{fig:steps_cb_cube_simplex_randh} that it grows as $O(d\lg d)$ which agrees with the analysis in Proposition~\ref{nballs2}. Then, due to the error splitting of Section~\ref{subsubsec:error_splitting} the generated points per volume ratio grows as $O^*(d)$. Thus, the total number of \billiard\ points grows as $O^*(d^2)$ for those polytopes (left plot in Figures~\ref{fig:steps_cb_cube_simplex_randh}). 
Thus, the total run-time of our implementation grows as $O(d^3m)$, because the per step cost of \billiard\ is $O(md)$.

For the rounding method of Section~\ref{sec:rounding} we experimentally evaluate its quality by counting the number of generated bodies in MMC. In particular, given a skinny polytope $P$ we round it and then we compare the number of bodies in MMC with that of isotropic polytopes (e.g.\ cube-$d$, $\Delta$-$d$-$d$). In Figure~\ref{fig:steps_cb_cube_simplex_randh}~(right) for  rhs-$d$-$m$ and $\mathcal{B}_n$ we notice that for both rounded and isotropic bodies the number of bodies grows as $O(d\lg d)$ which is a strong evidence that our rounding method transforms the polytope to a near isotropic position.

\begin{figure}[!h]
\begin{minipage}[h]{0.47\textwidth}
\includegraphics[width=\linewidth]{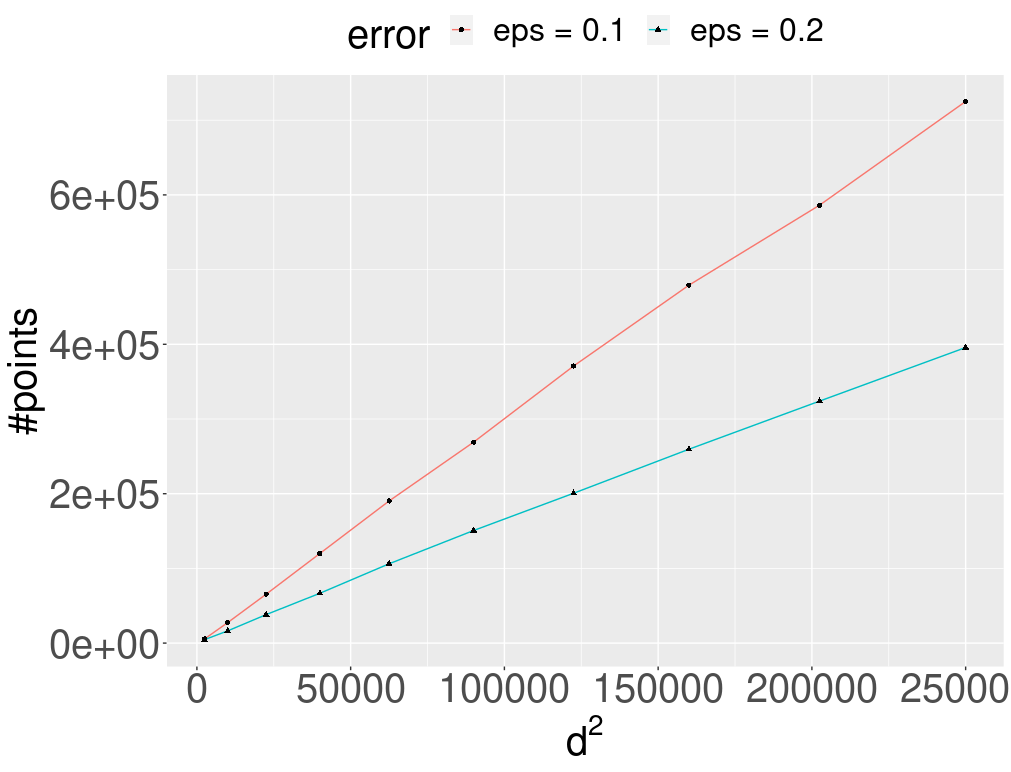}
\end{minipage}
\begin{minipage}[h]{0.47\textwidth}
\includegraphics[width=\linewidth]{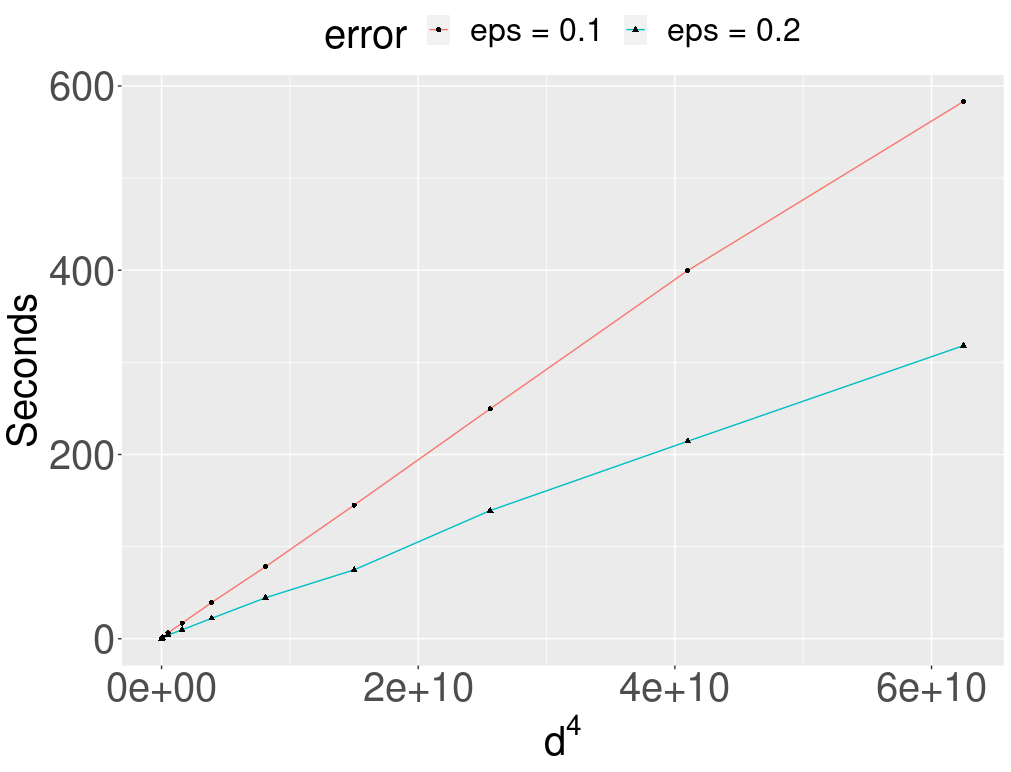}
\end{minipage}
\caption{Perfomance for different error parameters. \label{fig:steps_cb_cube_errors}}
\end{figure}

Next, we experimentally study how the complexity of our implementation depends on the input error parameter. In Figure~\ref{fig:steps_cb_cube_errors} we estimate the volumes of cube-$d$ for $d = 50,100,\dots, 500$ for two different values for the error parameter $\epsilon = 0.1,\ 0.2$. In the left plot we notice that for both values the number of generated points by \volalg\ increase as $O(d^2)$, while a smaller value of $\epsilon$ just increases the slope of the line, i.e.\ increases the run-time by a constant. The right plot also confirms this observation; it shows that the run-time of \volalg\ increases as $O(d^4)$ for cube-$d$. This is expected as the number of points increases as $O(d^2)$ and the per step cost of \billiard\ is $O(md)$, while $m=O(d)$ for cube-$d$.

Finally we examine the time spend in different steps in the algorithm i.e.\ rounding, constructing the sequence and volume computation. We conclude that those are highly depended on the convex body. For example for a 100-cube the times for the sequence construction and volume computation is 2.06836 and 4.29843 secs respectively. While for a 100-simplex those numbers are 2.85144 and 17.2575. Intuitively this effect is related to the isotropy of the convex set. 
Regarding preprocessing/rounding this also depends on the geometry of the body. For example, the rounding of the unit 100-cube takes 0.01427 secs while for a skinny 100-cube takes 6.92041 secs. 

\subsubsection{Comparison with Cooling Gaussian}

We compare our implementation, i.e.\ \volalg, against the state-of-the-art MATLAB implementation of \cg\ in~\cite{Cousins15} for cube-$d$, simplices $\Delta$-d and random polytopes rhs-$d$-$m$ with $m=5d$. 
In Figure~\ref{fig:steps_comparison}, we report the total number of generated points for both algorithms to achieve a relative error $\leq 0.1$. 
\cg\ performs computations up to $d=500$ while \volalg\ compute up to $d=1000$ in a shorter time frame. 
More interestingly, the ratio between the number of points generated by \cg\ over the number of points generated by \volalg\ increases linearly with the dimension. Then, the comparison of the run-times is determined by the per step cost of the random walks used in each implementation. 

The plot in Figure~\ref{fig:cb_cg_runtime_comparison} 
compares the run-times of both implementations for cube-$d$, $\Delta$-$d$ in H-representation (left plot), and rhs-$d$-$m$ with $m=5d$ and $d=50,100,\dots ,1000$. It confirms that \volalg\ is faster than \cg\ in~\cite{Cousins15}. Moreover, notice that the gap on the run-time increases with the dimension as expected. In particular, for $d=50$ our implementation is $20$, $2$ and $3$ times faster for cube-$d$, $\Delta$-$d$ and rhs-$d$-$m$ respectively. For $d=500$ our implementation is $95$, $8$ and $10$ times faster for the same classes of polytopes. For $d=1000$ we use polynomial interpolation to estimate the runtime of \cg, and our implementation is expected to be around $130$, $13$ and $20$ times faster for cube-$d$, $\Delta$-$d$ and rhs-$d$-$m$ respectively.

\begin{figure}[!h]
\centering
\includegraphics[width=0.49\textwidth]{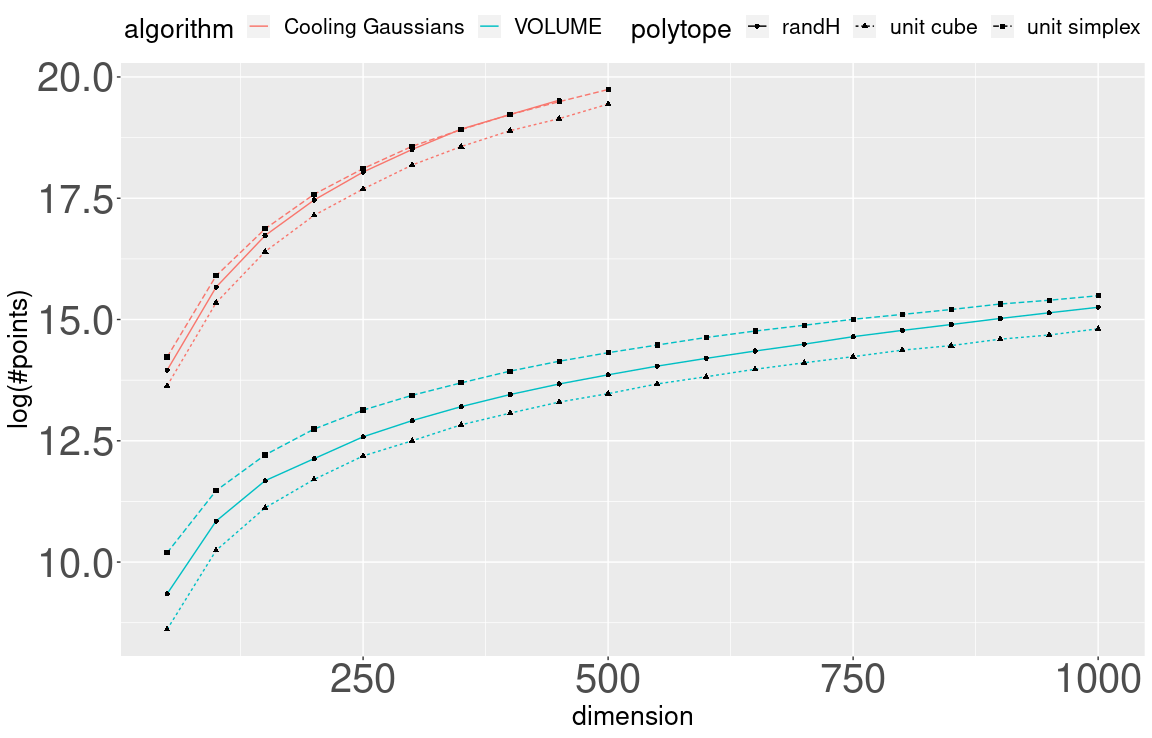}
\includegraphics[width=0.49\textwidth]{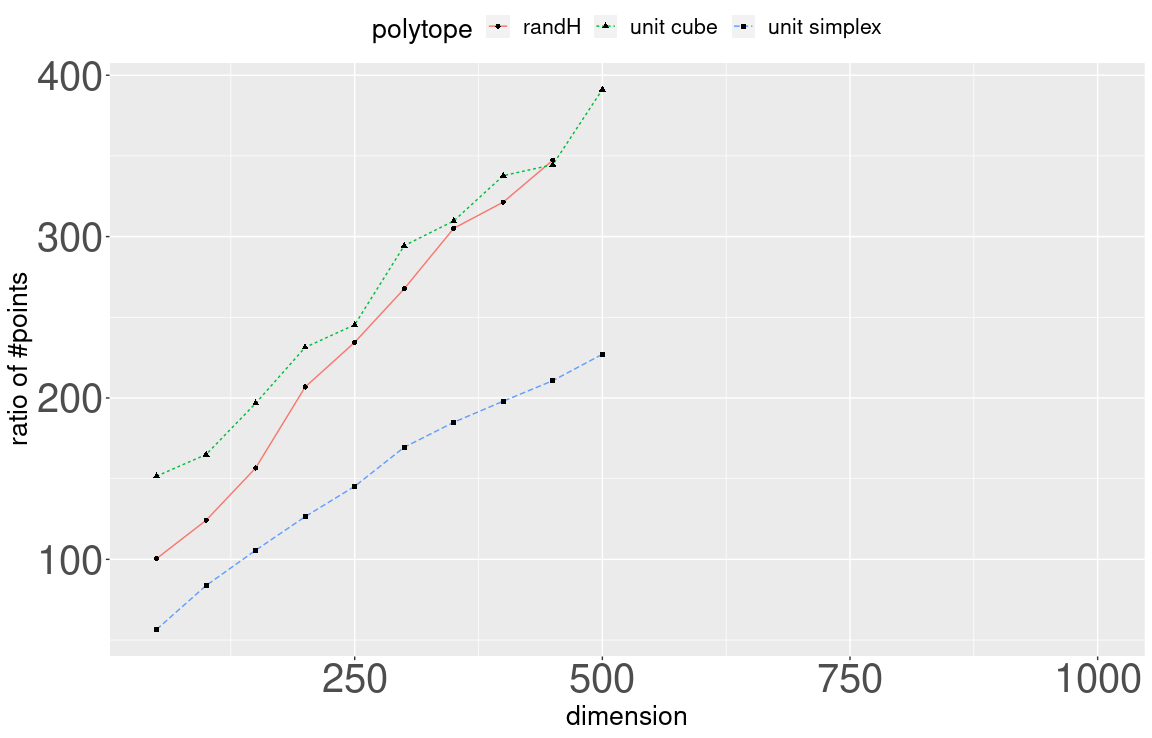}
\caption{Comparison of \volalg\ vs. \cg, i.e.\ our implementation compared the implementation in~\cite{Cousins15}. We consider the cases of randomly rotated cube-$d$, $\Delta$-$d$ in H-representation, and rhs-$d$-$m$ with $m=5d$ and $d=50,100,\dots ,1000$. 
For each class of polytopes we set the 12 hours as a time limit for the run-time of both algorithms. The runtime of \volalg\ never exceeds the 5 hours for all instances.
The average number of generated points by both implementations in log-scale (left) and the ratio of those numbers (right). In those experiments, we set the error parameter $\epsilon = 0.1$ for both implementations as input for the computations.
\label{fig:steps_comparison}}
\end{figure}

\begin{figure}[!h]
\centering
\includegraphics[width=0.49\textwidth]{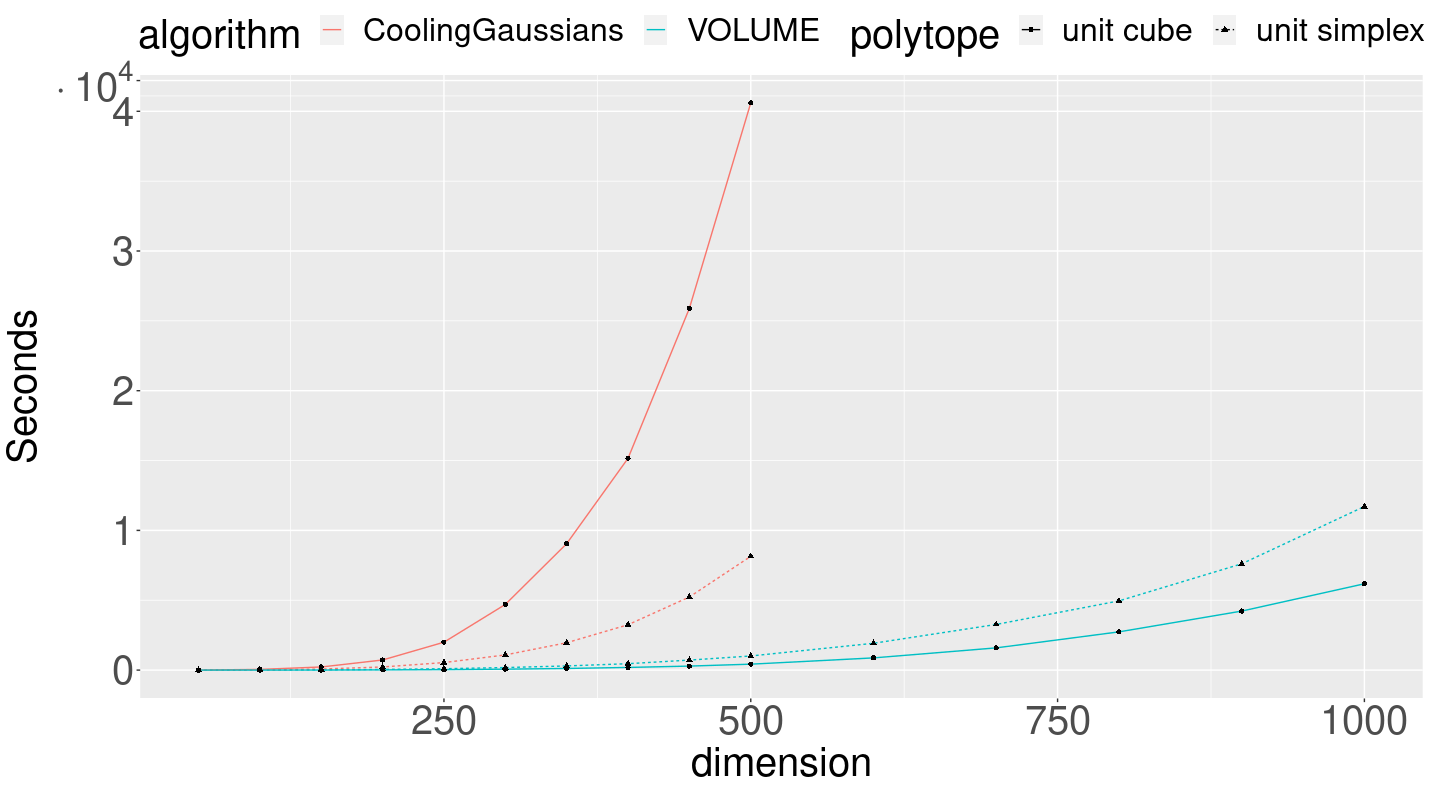}
\includegraphics[width=0.49\textwidth]{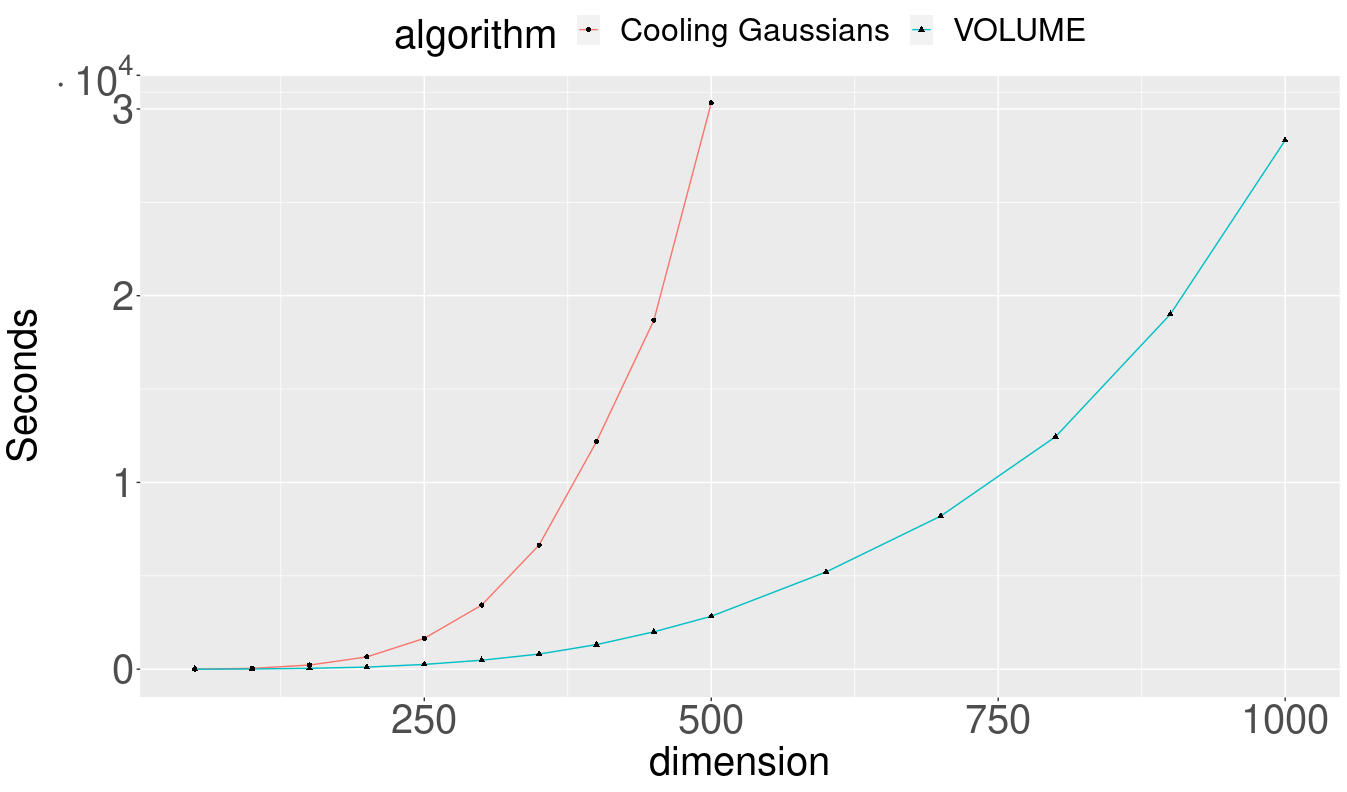}
\caption{Comparison of average run-time of \volalg\ vs. \cg, i.e.\ our implementation compared the implementation in~\cite{Cousins15}. We consider the case of randomly rotated cube-$d$, $\Delta$-$d$ in H-representation (left plot), and rhs-$d$-$m$ with $m=5d$ and $d=50,100,\dots ,1000$ (right plot).  
In those experiments, we set the error parameter $\epsilon = 0.1$ for both implementations as input for the computations.
\label{fig:cb_cg_runtime_comparison}}
\end{figure}

\subsection{Volume experiments for Z-polytopes}\label{sec:zpoly}


Our implementation is the first one that scales up to hundreds of dimensions in order of hours for Z-polytopes. 
In Table~\ref{fig:zono} we report the runtimes for both high and low order Z-polytopes. 
Our implementation, for a $100$-dimensional Z-polytope of order equal to 2 
performs $6.37\cdot 10^4$ reflections (equivalent to boundary oracle calls in \cite{Cousins15}) in less than an hour. 
Considering high order Z-polytopes our implementation scales up to order $150$ for $d\leq 15$ and to order $100$ and $20$ for $d=20$ and $d=30$ respectively in at most 
half an hour by performing $5.56\cdot 10^3$ reflections. 

The only alternative implementation for Z-polytopes is that in~\cite{Cousins15}, however they estimate volumes only for low dimensional Z-polytopes, while our implementation is undoubtedly superior. 
Since, we were unable to reproduce the results reported on~\cite{Cousins15} for Z-polytopes we compare against the data given in~\cite{Cousins15}. 
In particular, for a 2-order, $10$-dimensional Z-polytope our implementation performs $9.79\cdot 10^3$ reflections (boundary oracle calls) in $2.7$ seconds to achieve a relative error smaller than $0.1$; the implementation of \cg\ in~\cite{Cousins15} generates $9.58\cdot 10^4$  CDHR points which corresponds to $1.92\cdot 10^5$ boundary oracle calls in $3010$ seconds to achieve a relative error higher than $0.1$ and smaller than $0.2$.
For a $100$-order, $10$-dimensional Z-polytope our implementation performs $2.95\cdot 10^3$ reflections in $180$ seconds while the \cg\ in~\cite{Cousins15} performs $2.50\cdot 10^5$ boundary oracle calls in $4900$ seconds for the same values of relative errors. Clearly, for those instances, our implementation achieves a better accuracy, while the run-time is at least two order of magnitudes smaller than that of \cg\ in~\cite{Cousins15}.

Exact volume computation of Z-polytopes~\cite{volZono} require an exponential to the dimension $d$ number of operations. 
Thus, it can not scale beyond $d = 10$ for low order Z-polytopes. 

\begin{figure}[t]
\begin{minipage}[h]{0.45\textwidth}
\includegraphics[width=\linewidth]{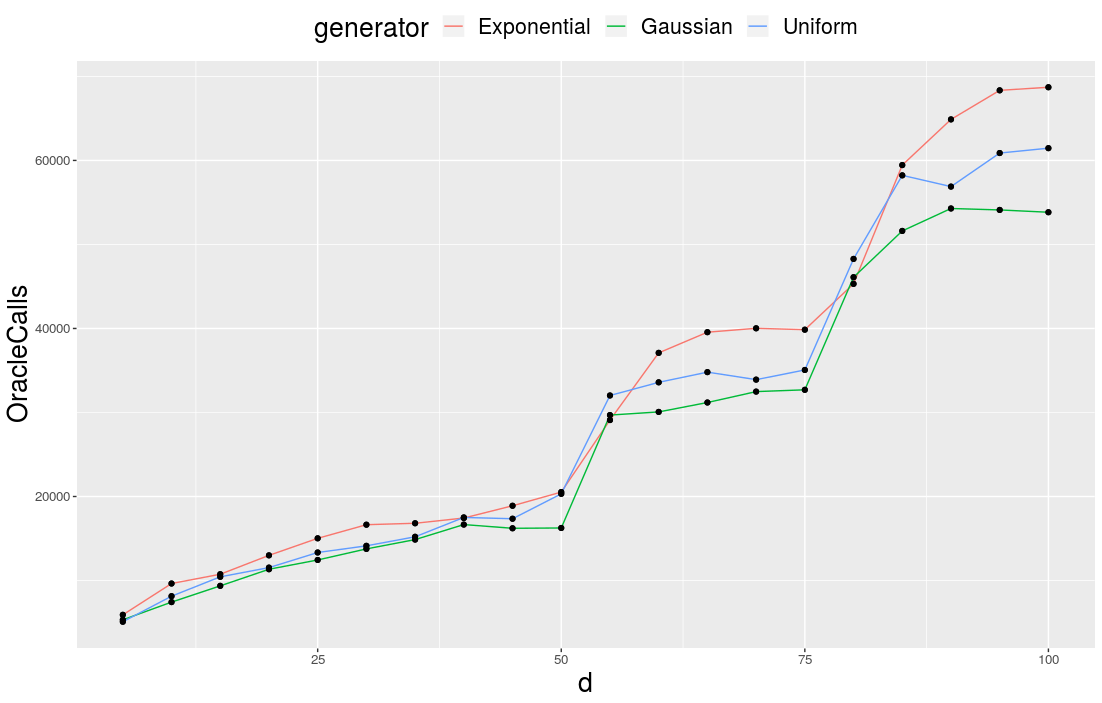}
\end{minipage}
\begin{minipage}[h]{0.45\textwidth}
\includegraphics[width=\linewidth]{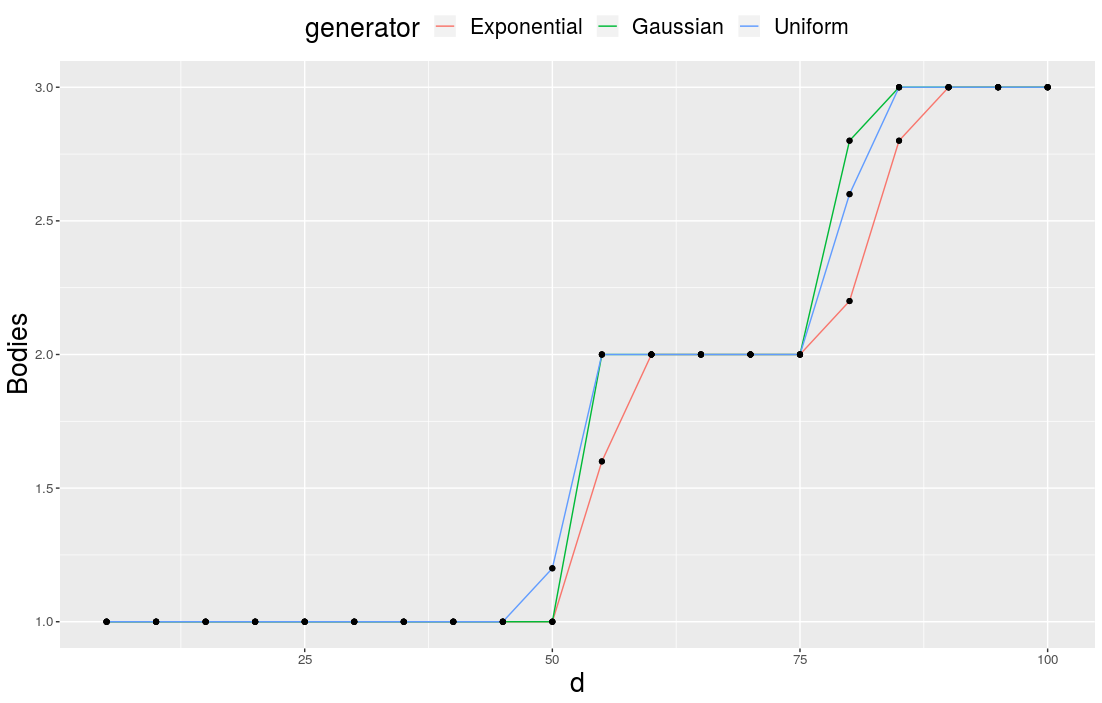}
\end{minipage}
\begin{minipage}[h]{0.45\textwidth}
\includegraphics[width=\linewidth]{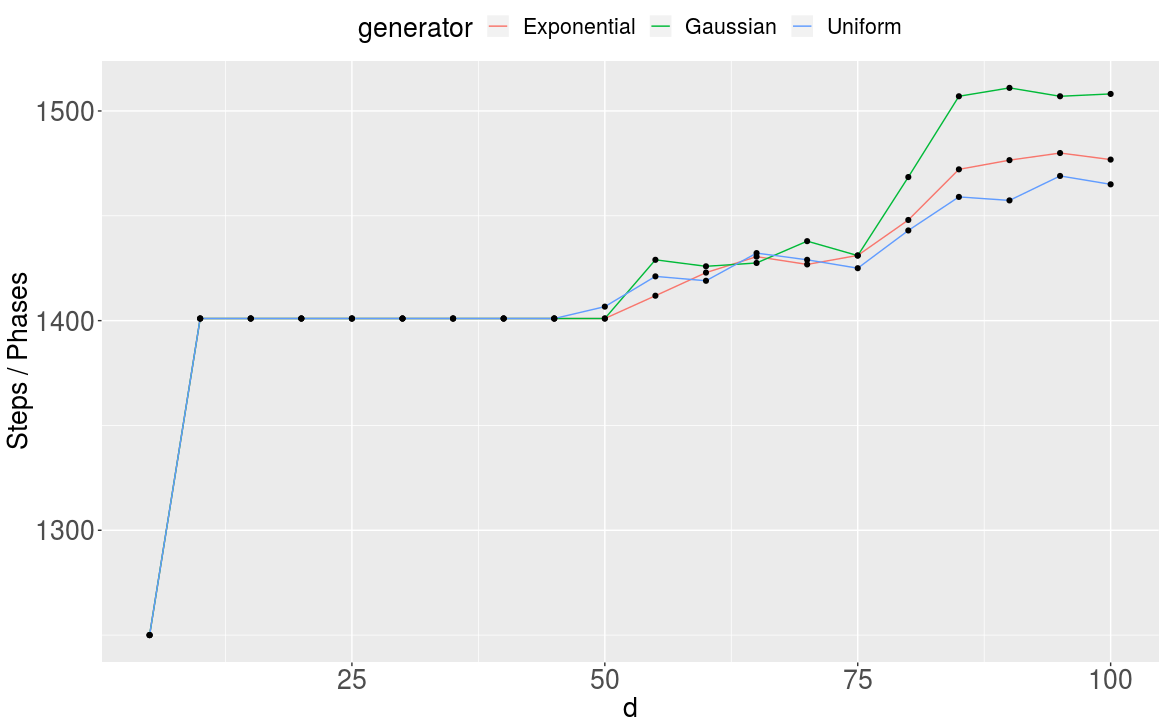}
\end{minipage}
\begin{minipage}[h]{0.45\textwidth}
\includegraphics[width=\linewidth]{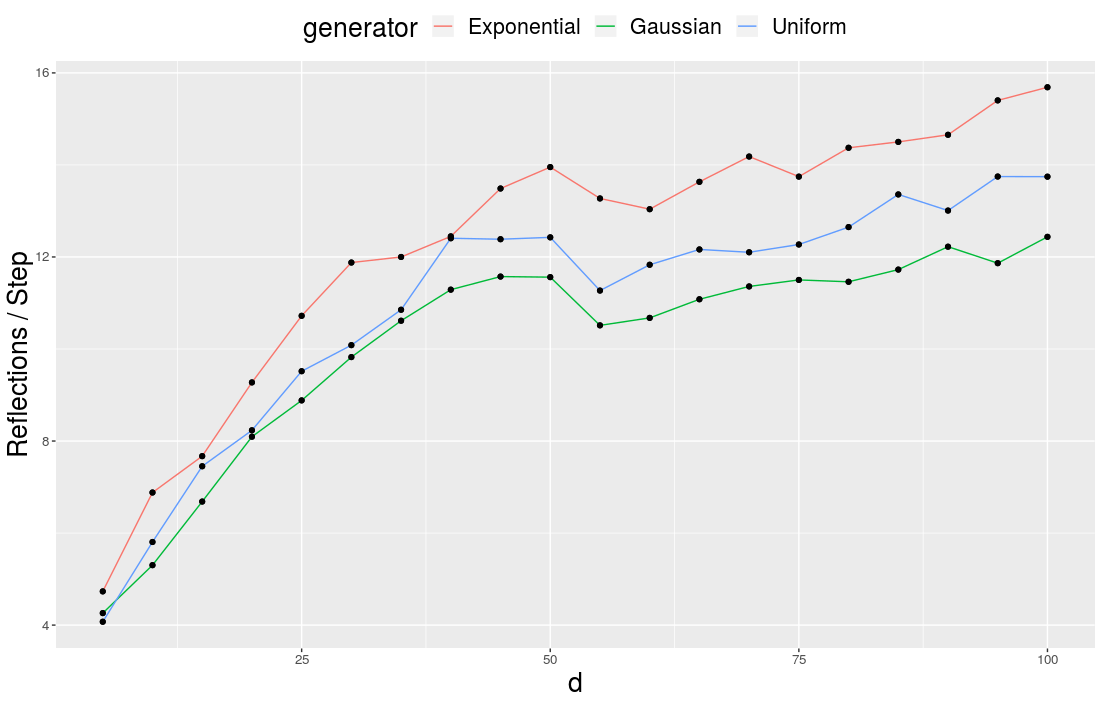}
\end{minipage}
\caption{Experimental complexity of \volalg\ for Z-polytopes of order 2. Total number of oracle calls (or reflections) is given by the ($\#$bodies in MMC)$\times$($\#$generated points per volume ratio)$\times$($\#$reflections per generated point). We set the error parameter $\epsilon = 0.1$ as input for the computations. \label{fig:zono_compl}}
\end{figure}

In the second column of Table~\ref{fig:zono} we report the type of body we use in MMC. 
For low order Z-polytopes we use the H-polytope from Section~\ref{sec:rounding} while for high order Z-polytopes we use the unit ball. 
Our experiments suggest to set the threshold for the order $n/d\leq 5$ so that the Z-polytope is considered as a low order one. 
In particular, Figure~\ref{fig:zono_bodies} shows that for order $\leq 4$, if we use the H-polytope in MMC, we get a bound on the the number of bodies in MMC $k\leq3$ for $d\leq 100$. 
Note that when we use the H-polytope the number of bodies $k$ is smaller for all pairs $(d,n)$ compared to using the unit ball without applying rounding to the Z-polytope. 
When we use balls in MMC, $n$ decreases for constant $d$ as $n$ increases. 
For order equals to $5$ the number of balls in MMC, without rounding, is equal or smaller than the number of H-polytopes in MMC when we use the H-polytope of Section~\ref{sec:rounding}. Table~\ref{fig:zono} shows that, for high-order Z-polytopes, $k=1$, which implies one or two acceptance-rejection steps. 
Moreover, Table~\ref{tab:r_nr_hp2} illustrates that the rounding method combined with balls in MMC results to a larger number of bodies in MMC and runtime than the case of using the H-polytope in MMC, while for order equal to $5$ the best runtime occurs when using a ball in MMC without rounding. 
Thus, we use the H-polytope in MMC if the order $n/d <5$. 
We conclude that for a Z-polytope of any order, no rounding preprocessing is needed, if we make the right choice of $C$ depending on the order; the maximum number of bodies is $k\leq3$ for $d\leq 100$.

To evaluate the efficiency of our implementation for Z-polytopes we wish to count the average number of reflections. 
We run our implementation of \volalg\ for $Z_{\mathcal{U}}$-$d$-$n$, $Z_{\mathcal{N}}$-$d$-$n$ and $Z_{Exp}$-$d$-$n$ Z-polytopes. An example for Z-polytopes of order $n/d=2$ is illustrated in Figure~\ref{fig:zono_compl}. First, as mentioned above our experiments imply that the number of bodies in MMC is $\leq 3$ for any order for $d\leq 100$. Moreover, the average number of reflections per point grows sub-linearly in $d$ (down-right plot in Figure~\ref{fig:zono_compl}). Considering that \billiard\ generates $O^*(1)$ points per volume ratio, the total number of reflections grows sub-linearly in $d$. 

\begin{figure}[t]
\begin{minipage}[h]{0.45\textwidth}
\includegraphics[width=\linewidth]{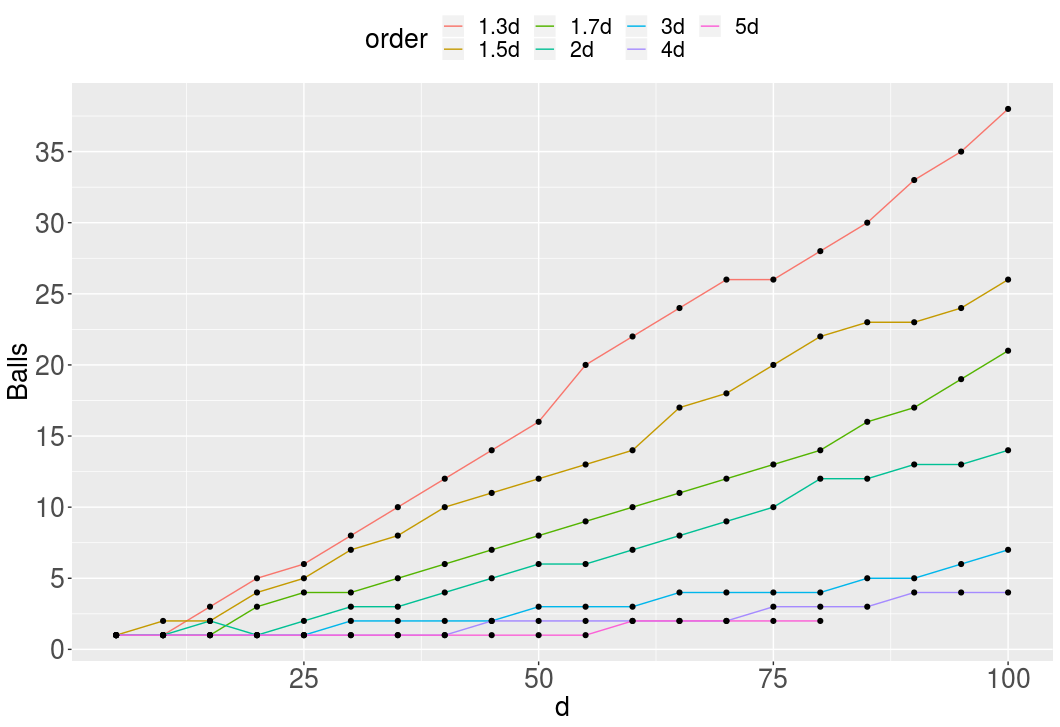}
\end{minipage}
\begin{minipage}[h]{0.45\textwidth}
\includegraphics[width=\linewidth]{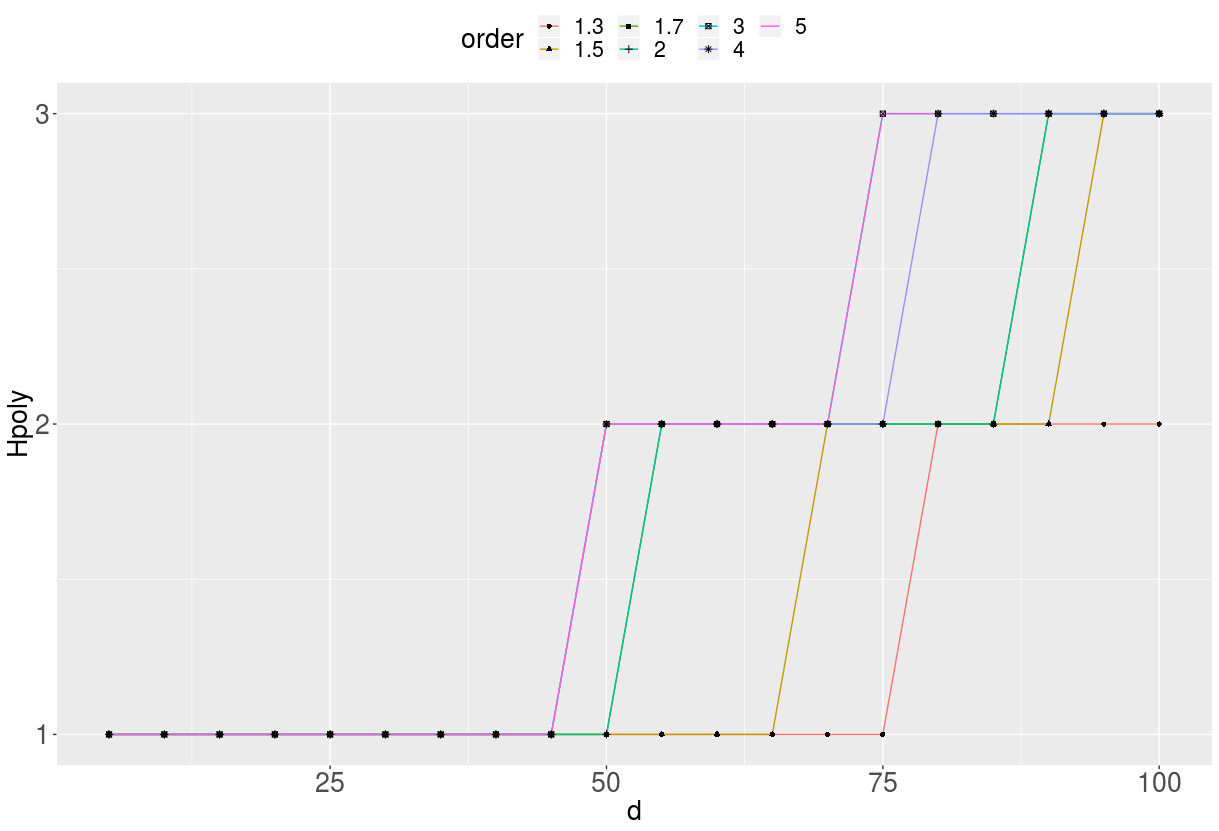}
\end{minipage}
\caption{Number of bodies in MMC for $Z_{\mathcal{U}}$. For each dimension we generate $10$ random Z-polytopes and we compute the number of bodies in MMC when $C$ is ball. We keep the Z-polytope with the larger number of bodies in MMC and then, for that one, we compute the number of bodies in MMC when $C$ is the Z-approximation.\label{fig:zono_bodies}}
\end{figure}

\subsubsection*{Application: Evaluate Z-polytope approximation} \label{zono_approx}

We propose an efficient algorithm for evaluating an over-approximation of a given Z-polytope $P$. Z-polytopes are critical in applications such as autonomous driving \cite{autdriv} or human-robot collaboration \cite{robcol}. Complexity strongly depends on the order of the encountered Z-polytopes. Thus, a practical solution is to over-approximate $P$, as tight as possible, with another Z-polytope $P_{red}\supseteq P$ of smaller order. 
A good measure of the approximation quality (fitness) is
\begin{equation}\label{pcafit}
R= (\vol(P_{red})/\vol(P))^{1/d} 
\end{equation}
Thus, to compute the approximation quality we need to compute volumes of Z-polytopes. In~\cite{Kopetzki17} they compute volumes exactly and deterministically, therefore it is impossible to compute the quality of approximation for $d>10$. 

Here, we exploit our software to test the quality of such approximations. The two methods that are able to scale for $d\geq 20$ are, primarily, the Principal Component Analysis (PCA) and the BOX method from~\cite{Kopetzki17}.
Both adopt similar approximations and are of comparable reliability; here we focus on PCA:
Let $X= [G | -G]^T$ and $\text{SVD}(X^T X)=U\Sigma V^T$, then
$G_{red}=U\cdot IH(U^T G)$ is square and generates $P_{red}$, where
$IH(\cdot)$ is the ``interval hull'' from \cite{IH}.  
Over-approximation can be seen as a reduction problem, so that the covariance among the $d$ generators of $P_{red}$ must be null.

Table~\ref{fig:zono} shows experimental results for Z-polytopes up to $d=100$ and various orders. $\vol(P_{red})$ is obtained exactly by computing one determinant. 
For PCA over-approximations we show that $R$ increases as $d$ grows but the same does not occur for fixed $d$ as order increases.
To our knowledge this is the first time practical volume estimation is used to test approximation methods in that high dimensional spaces. 

\begin{figure}[t]
\begin{minipage}[h]{0.45\textwidth}
\includegraphics[width=\linewidth]{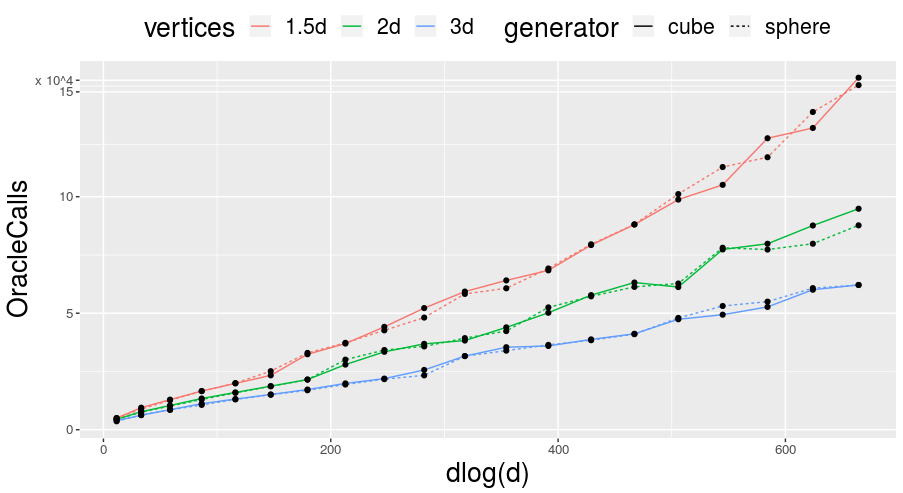}
\end{minipage}
\begin{minipage}[h]{0.45\textwidth}
\includegraphics[width=\linewidth]{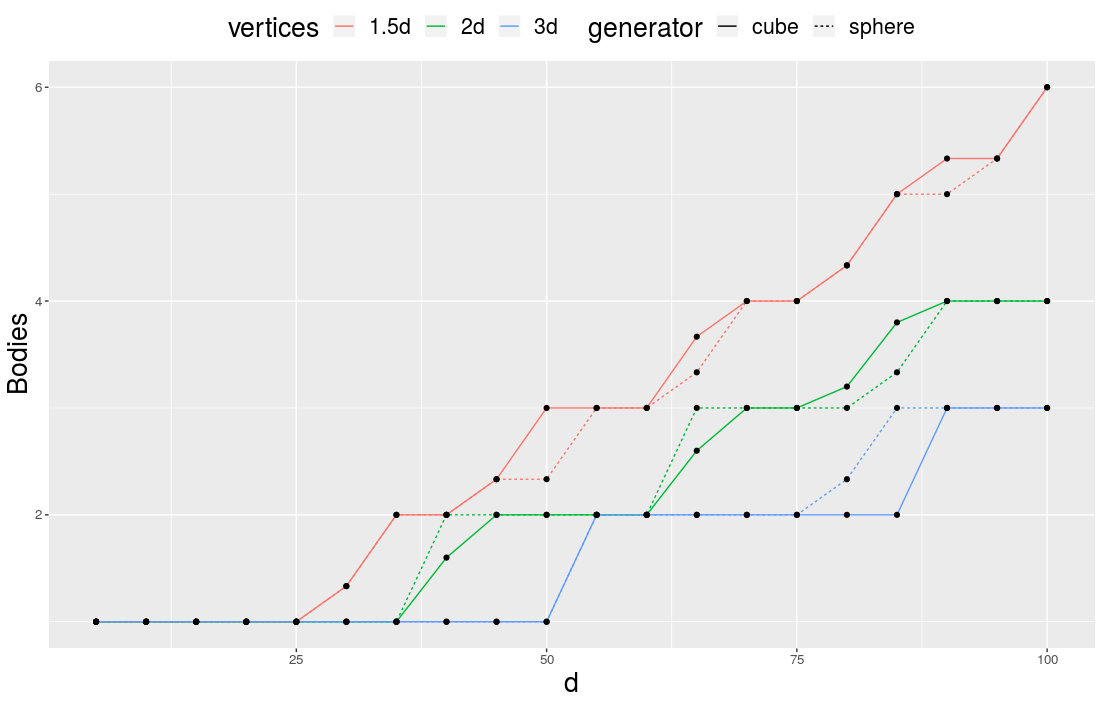}
\end{minipage}
\begin{minipage}[h]{0.45\textwidth}
\includegraphics[width=\linewidth]{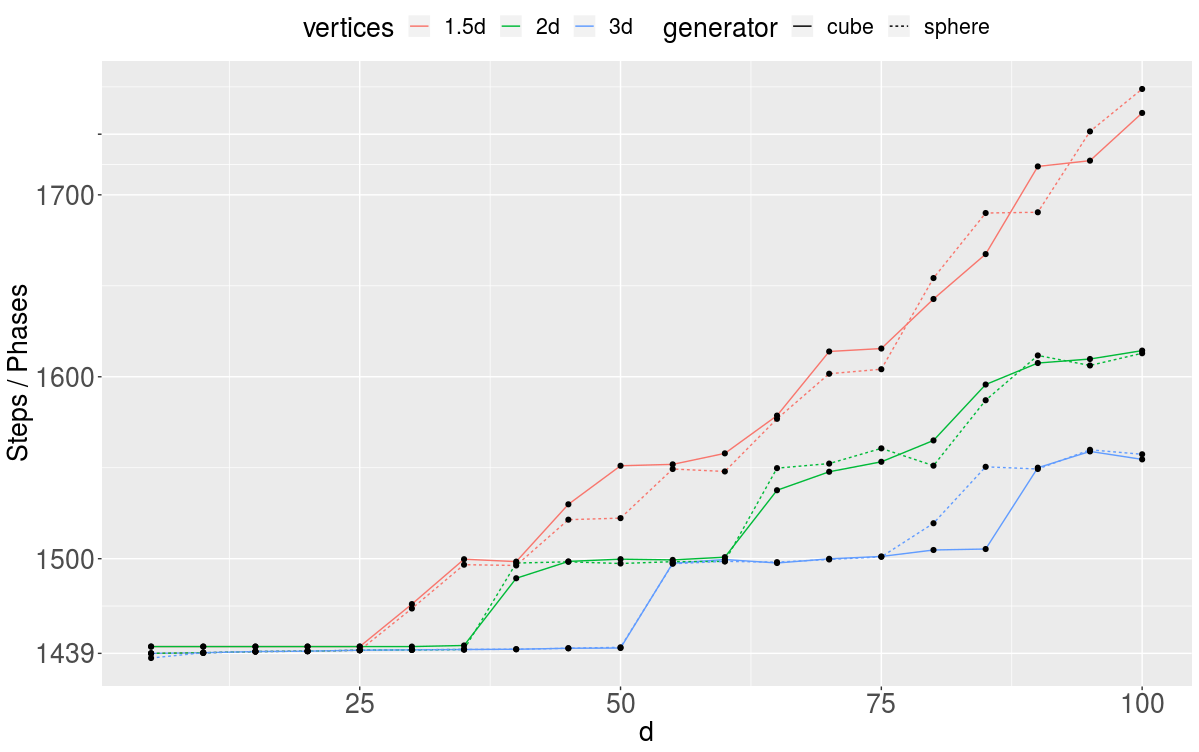}
\end{minipage}
\begin{minipage}[h]{0.45\textwidth}
\includegraphics[width=\linewidth]{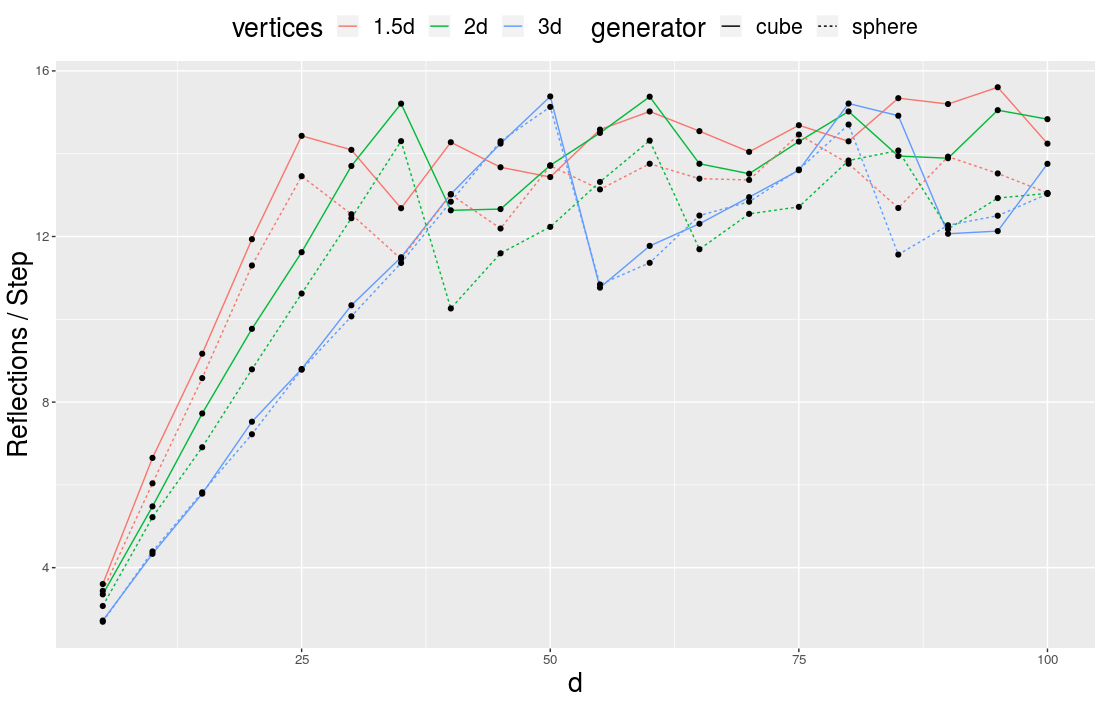}
\end{minipage}
\caption{Experimental complexity of \volalg\ for V-polytopes. Total number of oracle calls (or reflections) is given by the ($\#$bodies in MMC)$\times$($\#$generated points per volume ratio)$\times$($\#$reflections per generated point). We set the error parameter $\epsilon = 0.1$ as input for the computations. \label{fig:vpoly_compl}}
\end{figure}
\subsection{Volume experiments for V-polytopes}

In Table~\ref{fig:vpoly} we report the runtimes of our implementation for several V-polytopes. It estimates volume within a relative error at most $0.1$ in less than an hour. To round a V-polytope, we introduce a new method to control the sandwiching ratio $R/r$. 
Since the vertices are known, we put $P$ to John's position \cite{John48}.
First, we compute the smallest enclosing ellipsoid $E$ of $P$'s vertices by applying the approximation algorithm in~\cite{Todd07}. Notice that $E$ is also the smallest enclosing ellipsoid of $P$ and, if $E$ is a ball, then $P$ is in John's position. To round $P$, we apply to it the transformation that maps $E$ to the unit ball and repeat until the ratio between $E$'s longest and shortest axes falls under a given threshold.
Our experiments suggest this is much faster than the method in Section~\ref{sec:rounding}, as sampling is costly for this representation due to expensive boundary oracle calls. In the left plot of Figure~\ref{fig:vpoly_rounding}, notice that, for rounded V-polytopes, the number of bodies in MMC $k$ increases as a linear function of $d$. Moreover, as the number of vertices increases for constant $d$, then $k$ decreases. All rounding steps in our experiments for $d\leq 100$ took $\leq 3$ seconds.

To evaluate the performance of our algorithm we wish to count the average number of reflections (or boundary oracle calls). We run \volalg\ for rvc-$d$-$n$ and rvs-$d$-$n$ polytopes after a rounding step. An example of V-polytopes with number of vertices $n=1.5d,\ 2d,\ 3d$
is illustrated in Figure~\ref{fig:vpoly_compl}: The average number of reflections grows sublinearly in $d$ (down-right plot in Figure~\ref{fig:vpoly_compl}).
The total number of oracle calls grows as $O^*(d)$ (up-left plot in Figure~\ref{fig:vpoly_compl}). \volalg\ performs similarly for both rvc-$d$-$n$ and rvs-$d$-$n$ polytopes. 
However, the number of boundary oracle calls decreases as the number of vertices $n$ increases for constant $d$. 
This holds because after a rounding step the more the vertices of a random V-polytope the smaller the sandwiching ratio.

\section{Conclusions and future work}
We propose a new practical algorithm that computes for the first time volumes of very high-dimensional H-polytopes i.e.\ in the order of thousands of dimensions and high-dimensional V-, and Z- polytopes i.e.\ in the order of hundreds of dimensions. We provide strong empirical evidences on the efficiency of the algorithm and its accuracy through extended experiments on benchmark polytopes in computational geometry and bioinformatics.

We expect our software to address hard counting problems, e.g. counting the number of linear extensions of a partially ordered set~\cite{Talvitie18}. 
Our MMC scheme can be extended for other hard computational problems as Markov Chain Monte Carlo integration for multivariate integrals over a convex polytope. 
In particular, one could estimate the volume and generate the uniform samples required by one pass. 
Last but not least, we could exploit parallelism and faster software libraries for linear programs to scale to higher dimensions for V- and Z-polytopes.

\bibliographystyle{ACM-Reference-Format}
\bibliography{biblio} 

\begin{thebibliography}{10}

\bibitem{adamczak10}
Radosław Adamczak, Alexander Litvak, Alain Pajor, and Nicole
  Tomczak-Jaegermann.
\newblock Quantitative estimates of the convergence of the empirical covariance
  matrix in log-concave ensembles.
\newblock {\em J. of the American Mathematical Society}, 23(2):535--561, 2010.

\bibitem{autdriv}
M.~Althoff and J.M. Dolan.
\newblock Online verification of automated road vehicles using reachability
  analysis.
\newblock {\em IEEE Trans. Robotics}, 30(4):903--918, 2014.

\bibitem{Shiri20}
Shiri Artstein-Avidan, Haim Kaplan, and Micha Sharir.
\newblock On radial isotropic position: Theory and algorithms, 2020.

\bibitem{lpsolve}
M.~Berkelaar, K.~Eikland, and P.~Notebaert.
\newblock {\em {lp\_solve} 5.5, Open source (Mixed-Integer) Linear Programming
  system}, 2004.

\bibitem{Bertsimas04}
Dimitris Bertsimas and Santosh Vempala.
\newblock Solving convex programs by random walks.
\newblock {\em J. ACM}, 51(4):540–556, jul 2004.

\bibitem{Boyd04}
Stephen Boyd and Lieven Vandenberghe.
\newblock {\em Convex Optimization}.
\newblock Cambridge University Press, New York, NY, USA, 2004.

\bibitem{Brieden98}
A.~Brieden, P.~Gritzmann, R.~Kannan, V.~Klee, L.~Lov{\'{a}}sz, and
  M.~Simonovits.
\newblock Approximation of diameters: Randomization doesn't help.
\newblock In {\em 39th Symp. Foundations of Computer Science (FOCS)}, pages
  244--251, 1998.

\bibitem{brock2018large}
Andrew Brock, Jeff Donahue, and Karen Simonyan.
\newblock Large scale {GAN} training for high fidelity natural image synthesis.
\newblock In {\em 7th International Conference on Learning Representations,
  {ICLR} 2019, New Orleans, LA, USA, May 6-9, 2019}. OpenReview.net, 2019.

\bibitem{vinci}
B.~B\"ueler and A.~Enge.
\newblock {VINCI}, {2000}.
\newblock
  \url{http://www.math.u-bordeaux1.fr/~aenge/index.php?category=software\&page=vinci}.

\bibitem{Cales18}
L.~Cal\`es, A.~Chalkis, I.Z. Emiris, and V.~Fisikopoulos.
\newblock {Practical Volume Computation of Structured Convex Bodies, and an
  Application to Modeling Portfolio Dependencies and Financial Crises}.
\newblock In Bettina Speckmann and Csaba~D. T{\'o}th, editors, {\em 34th
  International Symposium on Computational Geometry (SoCG 2018)}, volume~99 of
  {\em LIPIcs}, pages 19:1--19:15, Dagstuhl, Germany, 2018. Schloss
  Dagstuhl--Leibniz-Zentrum fuer Informatik.

\bibitem{canfield07}
E.~Rodney Canfield and Brendan~D. McKay.
\newblock The asymptotic volume of the birkhoff polytope, 2007.

\bibitem{Cern12}
Michal Cern{\'y}.
\newblock Goffin's algorithm for zonotopes.
\newblock {\em Kybernetika}, 48:890--906, 2012.

\bibitem{Rvolesti}
Apostolos Chalkis and Vissarion Fisikopoulos.
\newblock {volesti: Volume Approximation and Sampling for Convex Polytopes in
  R}.
\newblock {\em {The R Journal}}, 13(2):642--660, 2021.

\bibitem{chalkis2020geometric}
Apostolos Chalkis, Vissarion Fisikopoulos, Elias Tsigaridas, and Haris
  Zafeiropoulos.
\newblock Geometric algorithms for sampling the flux space of metabolic
  networks.
\newblock In {\em International Symposium on Computational Geometry, SoCG},
  LIPIcs. Schloss Dagstuhl - Leibniz-Zentrum f{\"{u}}r Informatik, 2021.
\newblock to appear.

\bibitem{chen20}
Yuansi Chen.
\newblock An almost constant lower bound of the isoperimetric coefficient in
  the kls conjecture.
\newblock {\em Geom. Funct. Anal.}, 31:34--61, 2021.

\bibitem{CousinsV14}
B.~Cousins and S.~Vempala.
\newblock Bypassing {KLS: Gaussian} cooling and an ${O}^*(n^3)$ volume
  algorithm.
\newblock In {\em Proc. ACM STOC}, pages 539--548, 2015.

\bibitem{Cousins15}
B.~Cousins and S.~Vempala.
\newblock A practical volume algorithm.
\newblock {\em Mathematical Programming Computation}, 8, 2016.

\bibitem{Cramer46}
H.~Cramer.
\newblock {\em Mathematical methods of statistics}.
\newblock Princeton University Press, 1946.

\bibitem{Shcherbakov10}
F.~Dabbene, P.~S. Shcherbakov, and B.~T. Polyak.
\newblock A randomized cutting plane method with probabilistic geometric
  convergence.
\newblock {\em SIAM Journal on Optimization}, 20(6):3185--3207, 2010.

\bibitem{devroye_random_1984}
L~Devroye.
\newblock Random variate generation for unimodal and monotone densities.
\newblock {\em Computing}, 32(1):43--68, April 1984.

\bibitem{Dyer88}
M.~Dyer and A.~Frieze.
\newblock On the complexity of computing the volume of a polyhedron.
\newblock {\em SIAM Journal on Computing}, 17(5):967--974, 1988.

\bibitem{DyerFrKa91}
M.~Dyer, A.~Frieze, and R.~Kannan.
\newblock A random polynomial-time algorithm for approximating the volume of
  convex bodies.
\newblock {\em J. ACM}, 38(1):1--17, 1991.

\bibitem{Elekes1986}
G.~Elekes.
\newblock A geometric inequality and the complexity of computing volume.
\newblock {\em Discr.\ Comput.\ Geom.}, 1(4):289--292, 1986.

\bibitem{VolEsti}
{I. Z.} Emiris and V.~Fisikopoulos.
\newblock Practical polytope volume approximation.
\newblock {\em ACM Trans. Math. Soft.}, 44(4):38:1--38:21, 2018.
\newblock Prelim. version: Proc. SoCG 2014.

\bibitem{gamerman2006markov}
Dani Gamerman and Hedibert~F Lopes.
\newblock {\em Markov chain Monte Carlo: stochastic simulation for Bayesian
  inference}.
\newblock CRC Press, 2006.

\bibitem{Ge2015}
C.~Ge and F.~Ma.
\newblock A fast and practical method to estimate volumes of convex polytopes.
\newblock In J.~Wang and C.~Yap, editors, {\em Frontiers in Algorithmics},
  pages 52--65. Springer, 2015.

\bibitem{geyer_markov_1991}
Charles~J. Geyer.
\newblock Markov {Chain} {Monte} {Carlo} {Maximum} {Likelihood}.
\newblock Interface Foundation of North America, 1991.
\newblock Accepted: 2010-02-24T20:38:06Z.

\bibitem{volZono}
E.~Gover and N.~Krikorian.
\newblock Determinants and the volumes of parallelotopes and zonotopes.
\newblock {\em Linear Algebra and its Applications}, 413:28--40, 2010.

\bibitem{eigenweb}
Ga\"{e}l Guennebaud, Beno\^{i}t Jacob, et~al.
\newblock {\em {Eigen} v3}, 2010.

\bibitem{CousinsChnr}
S.~Haraldsdóttir, B.~Cousins, I.~Thiele, R.M.T Fleming, and S.~Vempala.
\newblock {CHRR: Coordinate Hit-and-run with rounding for uniform sampling of
  constraint-based models}.
\newblock {\em Bioinformatics}, 33(11):1741--1743, 2017.

\bibitem{he2017numerical}
Runxin He and Humberto Gonzalez.
\newblock Numerical synthesis of pontryagin optimal control minimizers using
  sampling-based methods.
\newblock In {\em 2017 IEEE 56th Annual Conference on Decision and Control
  (CDC)}, pages 733--738. IEEE, 2017.

\bibitem{huynh2012incremental}
Vu~Anh Huynh, Sertac Karaman, and Emilio Frazzoli.
\newblock An incremental sampling-based algorithm for stochastic optimal
  control.
\newblock In {\em 2012 IEEE International Conference on Robotics and
  Automation}, pages 2865--2872. IEEE, 2012.

\bibitem{Iyengar88}
S.~Iyengar.
\newblock Evaluation of normal probabilities of symmetric regions.
\newblock {\em SIAM Journal on Scientific and Statistical Computing},
  9(3):418--423, 1988.

\bibitem{Jia20}
H.~Jia, A.~Laddha, Y.T. Lee, and S.S. Vempala.
\newblock Reducing isotropy and volume to kls: An $o(n^3\psi^2)$ volume
  algorithm, 2020.

\bibitem{John48}
Fritz John.
\newblock Extremum {Problems} with {Inequalities} as {Subsidiary} {Conditions}.
\newblock In Giorgio Giorgi and Tinne~Hoff Kjeldsen, editors, {\em Traces and
  {Emergence} of {Nonlinear} {Programming}}, pages 197--215. Springer, Basel,
  2014.

\bibitem{Jolliffe16}
Ian~T. Jolliffe and Jorge Cadima.
\newblock Principal component analysis: a review and recent developments.
\newblock {\em Philosophical Transactions of the Royal Society A: Mathematical,
  Physical and Engineering Sciences}, 374(2065):20150202, 2016.

\bibitem{Kalai06}
Adam~Tauman Kalai and Santosh Vempala.
\newblock Simulated {Annealing} for {Convex} {Optimization}.
\newblock {\em Mathematics of Operations Research}, 31(2):253--266, 2006.
\newblock Publisher: INFORMS.

\bibitem{kaufman98}
D.E. Kaufman and R.L. Smith.
\newblock Direction {Choice} for {Accelerated} {Convergence} in {Hit}-and-{Run}
  {Sampling}.
\newblock {\em Operations Research}, 46(1):84--95, 1998.

\bibitem{Kerber17}
M.~Kerber, R.~Tichy, and {M.F.} Weitzer.
\newblock Constrained triangulations, volumes of polytopes, and unit equations.
\newblock In {\em 33rd Intern. Symp. Computational Geometry (SoCG 2017)}, pages
  46:1--46:15, Germany, 2017. Schloss Dagstuhl - Leibniz-Zentrum f{\"u}r
  Informatik.

\bibitem{kook22}
Yunbum Kook, Yin~Tat Lee, Ruoqi Shen, and Santosh~S. Vempala.
\newblock Sampling with riemannian hamiltonian monte carlo in a constrained
  space.
\newblock {\em CoRR}, abs/2202.01908, 2022.

\bibitem{Kopetzki17}
A.K. Kopetzki, B.~Sch\"urmann, and M.~Althoff.
\newblock Methods for order reduction of zonotopes.
\newblock In {\em Proc. IEEE Annual Conf. Decision \& Control (CDC)}, pages
  5626--5633, 2017.

\bibitem{NLoptSQP}
D.~Kraft.
\newblock A software package for sequential quadratic programming.
\newblock {\em Technical report, DFVLR-FB 88-28, Institut für Dynamik der
  Flugsysteme}, 1988.

\bibitem{IH}
W.~K{\"u}hn.
\newblock Rigorously computed orbits of dynamical systems without the wrapping
  effect.
\newblock {\em Computing}, 61:47--67, 1998.

\bibitem{Laddha20}
Aditi Laddha and Santosh Vempala.
\newblock {Convergence of Gibbs Sampling: Coordinate Hit-and-Run Mixes Fast},
  2020.

\bibitem{Lee18}
Y.T. Lee and S.S. Vempala.
\newblock Convergence rate of riemannian hamiltonian monte carlo and faster
  polytope volume computation.
\newblock In {\em Proceedings of the 50th Annual ACM SIGACT Symposium on Theory
  of Computing}, STOC 2018, page 1115–1121, New York, NY, USA, 2018.
  Association for Computing Machinery.

\bibitem{LovSim}
L.~Lov\'asz, R.~Kannan, and M.~Simonovits.
\newblock Random walks and an ${O}^*(n^5)$ volume algorithm for convex bodies.
\newblock {\em Random Structures and Algorithms}, 11:1--50, 1997.

\bibitem{Lovasz06}
L.~Lov\'{a}sz and S.~Vempala.
\newblock Hit-and-run from a corner.
\newblock {\em SIAM J. Comp.}, 35(4):985--1005, 2006.

\bibitem{LovVem}
L.~Lovász and S.~Vempala.
\newblock Simulated annealing in convex bodies and an {O}$^*(n^4)$ volume
  algorithms.
\newblock {\em J. Computer \& System Sciences}, 72:392--417, 2006.

\bibitem{Mangoubi19}
O.~{Mangoubi} and N.~K. {Vishnoi}.
\newblock Faster polytope rounding, sampling, and volume computation via a
  sub-linear ball walk.
\newblock In {\em 2019 IEEE 60th Annual Symposium on Foundations of Computer
  Science (FOCS)}, pages 1338--1357, 2019.

\bibitem{Martino16}
L.~Martino, V.~Elvira, D.~Luengo, J.~Corander, and F.~Louzada.
\newblock Orthogonal parallel mcmc methods for sampling and optimization.
\newblock {\em Digital Signal Processing}, 58:64–84, Nov 2016.

\bibitem{martino_independent_2018}
Luca Martino, David Luengo, and Joaquín Míguez.
\newblock Independent {Sampling} for {Multivariate} {Densities}.
\newblock In Luca Martino, David Luengo, and Joaquín Míguez, editors, {\em
  Independent {Random} {Sampling} {Methods}}, Statistics and {Computing}, pages
  197--247. Springer International Publishing, Cham, 2018.

\bibitem{Vpolyinsc}
K.G. Murty.
\newblock Ball centers of special polytopes.
\newblock Technical report, Department of Industrial \& Operations Engineering,
  University of Michigan, 2009.

\bibitem{Narayanan20}
Hariharan Narayanan and Piyush Srivastava.
\newblock On the mixing time of coordinate hit-and-run, 2020.

\bibitem{neal2011mcmc}
Radford~M Neal et~al.
\newblock {MCMC using Hamiltonian dynamics}.
\newblock {\em Handbook of markov chain monte carlo}, 2(11):2, 2011.

\bibitem{robcol}
A.~Pereira and M.~Althoff.
\newblock Safety control of robots under computed torque control using
  reachable sets.
\newblock In {\em Proc. IEEE Inter. Conf. Robotics Automation (ICRA)}, pages
  331--338, 2015.

\bibitem{POLYAK20146123}
B.T. Polyak and E.N. Gryazina.
\newblock Billiard walk - a new sampling algorithm for control and
  optimization.
\newblock {\em IFAC Proceedings Volumes}, 47(3):6123--6128, 2014.
\newblock 19th IFAC World Congress.

\bibitem{cobra20}
Chaitra Sarathy, Martina Kutmon, Michael Lenz, Michiel~E. Adriaens, Chris~T.
  Evelo, and Ilja~C.W. Arts.
\newblock Efmviz: A cobra toolbox extension to visualize elementary flux modes
  in genome-scale metabolic models.
\newblock {\em Metabolites}, 10(2), 2020.

\bibitem{Schellenberger09}
J.~Schellenberger and B.O. Palsson.
\newblock Use of randomized sampling for analysis of metabolic networks.
\newblock {\em J. Biological Chemistry}, 284(9):5457--5461, 2009.

\bibitem{schumer_adaptive_1968}
M.~Schumer and K.~Steiglitz.
\newblock Adaptive step size random search.
\newblock {\em IEEE Transactions on Automatic Control}, 13(3):270--276, June
  1968.
\newblock Conference Name: IEEE Transactions on Automatic Control.

\bibitem{Talvitie18}
T.~Talvitie, K.~Kangas, T.~Niinim\"aki, and M.~Koivisto.
\newblock Counting linear extensions in practice: Mcmc versus exponential monte
  carlo.
\newblock In {\em AAAI Conference on Artificial Intelligence}, 2018.

\bibitem{Todd07}
Michael~J. Todd and E.~Alper Yildirim.
\newblock On {Khachiyan's} algorithm for the computation of minimum-volume
  enclosing ellipsoids.
\newblock {\em Discrete Applied Mathematics}, 155(13):1731--1744, 2007.

\bibitem{venzke2019}
A.~Venzke, D.K. Molzahn, and S.~Chatzivasileiadis.
\newblock Efficient creation of datasets for data-driven power system
  applications.
\newblock {\em Electric Power Systems Research}, 190:106614, 2021.

\bibitem{Wong03}
Frederick Wong, Christopher~K. Carter, and Robert Kohn.
\newblock Efficient estimation of covariance selection models.
\newblock {\em Biometrika}, 90(4):809--830, 2003.

\end{thebibliography}

\newpage
\appendix

\section{Figures and tables}

\begin{table}[ht]
\centering
\begin{tabular}{|l||c|c|c|c|c|r|}\hline
 H-polytope & volume & std/mean & error & $k$ & points & time  \\ \hline\hline
cross-$10$ & 2.81e-04 & 0.06 & 0.01 & 1.0 & 1.50e+03 & 0.11 \\
\hline
cross-$13$ & 1.30e-06 & 0.07 & 0.01 & 1.0 & 1.50e+03 & 1.17 \\
\hline
cross-$15$ & 2.49e-08 & 0.04 & 0.01 & 1.0 & 1.50e+03 & 12.8 \\
\hline
cube-$100$ & 1.18e+30 & 0.15 & 0.07 & 9.2 & 2.80e+04 & 1.5 \\
\hline
cube-$250$ & 1.96e+75 & 0.16 & 0.08 & 34.5 & 1.96e+05 & 36.2 \\
\hline
cube-$500$ & 3.38e+150 & 0.25 & 0.03 & 84.7 & 7.09e+05  & 425 \\
\hline
cube-$750$ & 5.75e+225 & 0.22 & 0.03 & 139.9 & 1.52e+06 & 1501 \\
\hline
cube-$1000$ & 2.11e+301 & 0.29 & 0.03 & 201.0 & 2.70e+06 & 6180 \\
\hline
$\Delta$-$100$ & 1.02e-158 & 0.16 & 0.05 & 21.0 & 1.01e+05  & 4.8 \\
\hline
$\Delta$-$250$ & 3.20e-493 & 0.15 & 0.03 & 66.3 &  5.05e+05 & 66.4 \\
\hline
$\Delta$-$500$ & 7.55e-1135 & 0.20 & 0.08 & 149.1 & 1.65e+06  & 847.8 \\
\hline
$\Delta$-$750$ & 3.70e-1833 & 0.17 & 0.04 & 237.9 & 3.28e+06  & $4088$ \\
\hline
$\Delta$-$1000$ & 2.67e-2568 & 0.24 & 0.07 & 328.4 &  5.34e+06 & $10711$ \\
\hline
$\Delta$-$50$-$50$ & 1.07e-129 & 0.21 & 0.01 & 22.3 & 9.90e+04 & 3.3 \\
\hline
$\Delta$-$125$-$125$ & 2.91e-419 & 0.18 & 0.03 & 67.0 & 5.07e+05 & 107.2  \\
\hline
$\Delta$-$250$-$250$ & 9.10e-986 & 0.23 & 0.05 & 151.3 & 1.69e+06 & 958.5 \\
\hline
$\Delta$-$375$-$375$ & 6.26e-1609 & 0.21 & 0.06 & 239.3 & 3.35e+06 & $4344$ \\
\hline
$\Delta$-$500$-$500$ & 6.22e-2269 & 0.24 & 0.07 & 332.5 & 5.39e+06 & $9432$ \\
\hline
$rhs$-$100$-$1000$ & 4.19e-46 & 0.11 & ?? & 10.6 & 3.33e+04 & 5.25 \\
\hline
$rhs$-$300$-$2000$ & 2.75e-205 & 0.20 & ?? & 50.0 & 3.22e+05 & 253 \\
\hline
$rhs$-$1000$-$6000$ & 4.18e+275 & 0.23 & ?? & 202.1 & 2.61e+06 & $19431$ \\
\hline
$rhs$-$1000$-$10000$ & 1.47e+216 & 0.21 & ?? & 172.3 & 1.90e+06 & $22833$ \\
\hline
\hline
e\_coli\_core [174-24] & 3.37e-47 & 0.08 & ?? & 2.0 & 3.52e+03 & 0.16 \\
 \hline
iAB\_RBC\_283 [906-130]  & 1.73e+43 &  0.28 & ?? & 23.9 & 1.11e+05 & 31.3 \\
\hline
iJR904 [1132-227] & 6.00e+53 & 0.25 & ?? & 50.8 & 3.38e+05 & 169.3 \\
\hline
iAT\_PLT\_636 [2016-289] & 3.65e-854 & 0.25 & ?? & 66.3 & 4.99e+05 & 512.7 \\
\hline
iSDY\_1059 [2966-509] & 2.25e-350 & 0.27 & ?? & 136.0 & 1.43e+06 & 3486 \\
\hline
Recon1 [4934-931] & 8.01e-5321 & 0.17 & ?? & 266.1 & 3.84e+06 & $18225$ \\
\hline
\end{tabular}
\caption{\label{tab:h_poly_vol} Volume estimation for H-polytopes by \volalg. 
For each polytope \volalg\ performs $10$ runs; 
volume: mean volume computed;
std/mean: ratio of std and mean of volumes;
error: the relative error; 
$k$: average number of generated bodies in MMC; 
points: average number of generated points; 
time: average time in seconds; 
$\epsilon = 0.1$ used in all cases;
??: the exact volume is unknown.
The last $6$ are metabolic polytopes with $m$ facets in $d$ dimensions denoted as $[m-d]$. 
For metabolic polytopes we follow the preprocessing of {\tt cobra} package~\cite{cobra20} also used in~\cite{chalkis2020geometric,CousinsChnr}.}
\end{table}

\begin{table}[ht]
\centering
\begin{tabular}{|l||c|c|c|c|c|c|c|c|r|}\hline
 $\mathcal{B}_n$ & $d$ & exact & asymptotic & volume & error & std/mean & $k$ & points & time  \\ \hline\hline
$\mathcal{B}_3$ & 4 & 1.12 & 1.410849 & 1.12e+00 & 0.00 &  0.02 & 1 & 1.50e+03 & 0.01 \\
\hline
$\mathcal{B}_4$ & 9 & 6.21e-02 & 7.61e-02 & 6.17e-02 & 0.01 & 0.02 & 1 & 1.50e+03 & 0.01 \\
\hline
$\mathcal{B}_5$ & 16 &  1.41e-04 & 1.69e-04 & 1.38e-04 & 0.02 & 0.01 & 1 & 1.52e+03  &  0.03 \\
\hline
$\mathcal{B}_6$ &  25 & 7.35e-09 & 8.62e-09 & 7.35e-09 & 0.00 & 0.05 & 2.3 & 3.63e+03 & 0.06 \\
\hline
$\mathcal{B}_7$ & 36 &  5.64e-15 & 6.51e-15 & 5.49e-15 & 0.03 & 0.02 & 4.0 & 7.98e+03 & 0.16 \\
\hline
$\mathcal{B}_8$ & 49 & 4.42e-23 & 5.03e-23 & 4.24e-23 & 0.04 & 0.05 & 6.4 & 1.62e+04 & 0.37 \\
\hline
$\mathcal{B}_9$ & 64 &  2.60e-33 & 2.93e-33 & 2.54e-33 & 0.02 & 0.03 & 9.2 &  2.82e+04  & 0.86 \\
\hline
$\mathcal{B}_{10}$ & 81 &  8.78e-46 & 9.81e-46 & 8.46e-46 & 0.04 & 0.05 & 13.1 &  4.69e+04 & 1.9 \\
\hline
$\mathcal{B}_{11}$ & 100 & ?? &  1.49e-60 & 1.24e-60 & ?? & 0.05 & 17.6 & 7.17e+04 & 3.5 \\
\hline
$\mathcal{B}_{12}$ & 121 & ?? &   8.38e-78 & 7.28e-78 & ?? & 0.10 & 23.1 &  1.07e+05 & 6.7 \\
\hline
$\mathcal{B}_{13}$ & 144 & ?? &   1.43e-97 & 1.205e-97 & ?? & 0.19 & 29.2 & 1.49e+05 & 11.9 \\
\hline
$\mathcal{B}_{14}$ & 169 & ?? &   6.24e-120 & 5.13e-120 & ?? & 0.08 & 35.9 & 2.04e+05  & 20.7 \\
\hline
$\mathcal{B}_{15}$ & 196 & ?? &   5.94e-145 & 4.92e-145 & ?? & 0.05 & 43.4 & 2.68e+05 & 34.5 \\
\hline
$\mathcal{B}_{16}$ & 225 & ?? &   1.06e-172 & 9.62e-173 & ?? & 0.12 & 51.9 & 3.49e+05 &  56.3 \\
\hline
$\mathcal{B}_{17}$ & 256 & ?? &   3.10e-203 & 2.50e-203 & ?? & 0.10 & 60.9 & 4.41e+05 & 84.2  \\
\hline
$\mathcal{B}_{18}$ & 289 & ?? &   1.30e-236 & 1.03e-237 & ?? & 0.10 & 72.0 & 5.65e+05 &  129 \\
\hline
$\mathcal{B}_{19}$ & 324 & ?? &   6.96e-273 & 5.41e-273 & ?? & 0.09 & 82.0 & 6.88e+05 & 198  \\
\hline
$\mathcal{B}_{20}$ & 361 & ?? &   4.23e-312 & 3.31e-312 & ?? & 0.11 & 94.9 & 8.46e+05 &  297 \\
\hline\hline
$\mathcal{B}_{21}$ & 400 & ?? & 2.62e-354 & 1.77e-354 & ?? & 0.27 & 101.6 & 9.40e+05 & 325  \\
\hline
$\mathcal{B}_{22}$ & 441 & ?? & 1.49e-399 & 1.49e-399 & ?? & 0.17 & 115.0 & 1.12e+06 &  453 \\
\hline
$\mathcal{B}_{23}$ & 484 & ?? & 7.10e-448 & 4.33e-448 & ?? & 0.30 & 128.8 & 1.33e+06 &  624 \\
\hline
$\mathcal{B}_{24}$ & 529 & ?? & 2.57e-499 & 2.10e-499 & ?? & 0.27 & 143.7 & 1.55e+06 &  845 \\
\hline
$\mathcal{B}_{25}$ & 576 & ?? & 6.46e-554 & 5.56e-554 & ?? & 0.28 & 158.8 & 1.82e+06 & 1060  \\
\hline
$\mathcal{B}_{26}$ & 625 & ?? & 1.04e-611 & 7.80e-612 & ?? & 0.22 & 176.3 & 2.11e+06 & 1405 \\
\hline
$\mathcal{B}_{27}$ & 676 & ?? & 9.81e-673 & 1.00e-672 & ?? & 0.30 & 194.0 & 2.41e+06 & 1918  \\
\hline
$\mathcal{B}_{28}$ & 729 & ?? & 5.05e-737 & 6.13e-737 & ?? & 0.31 & 211.1 & 2.75e+06 & 2573  \\
\hline
$\mathcal{B}_{29}$ & 784 & ?? & 1.31e-804 & 1.17e-804 & ?? & 0.31 &  231.5  & 3.15e+06 &  3539 \\
\hline
$\mathcal{B}_{30}$ & 841 & ?? & 1.60e-875 & 1.35e-875 & ?? & 0.35 & 249.1 & 3.53e+06 & 4982  \\
\hline
$\mathcal{B}_{31}$ & 900 & ?? & 8.55e-950 & 8.43e-950 & ?? & 0.40 & 272.9 & 4.02e+06 & 6967  \\
\hline
$\mathcal{B}_{32}$ & 961 & ?? & 1.86e-1027 & 1.58e-1027 & ?? & 0.39 & 293.8 & 4.47e+06 & 8905  \\
\hline
$\mathcal{B}_{33}$ & 1024 & ?? & 1.56e-1108 & 1.37e-1108 & ?? & 0.42 & 314.1 & 4.97e+06 &  11105  \\
\hline
\end{tabular}
\caption{\label{tab:birkhoff_vol} Average volume of Birkhoff polytopes ($\mathcal{B}_n$) over $50$ runs of \volalg\ for $n\leq 20$ and over $10$ runs for $21\leq n\leq 28$; 
$d$: the dimension of $\mathcal{B}_n$; 
exact: the exact volume; 
asymptotic: the asymptotic volume from~\cite{canfield07}; 
volume: mean estimated volume; 
error: relative error; 
std/mean: the ratio between std and mean of volumes; 
$k$: average number of bodies in MMC;
points: average number of generated points; 
time: average time in seconds; 
??: the exact volume is unknown;
$e=0.1$ in all cases.}
\end{table}

\begin{table}[ht]
\centering
\begin{tabular}{|l||c|c|c|c|c|c|r|c|c|}\hline
 Z-polytope &  Body & order &  volume & std/mean & $k$ & Refl  & time & $\vol(P_{red})$ & $R$  \\ \hline\hline
\centering $Z_{\mathcal{U}}$-$10$-$200$ & \centering Ball & \centering $20$ & 5.47e+34 & 0.06 & 1 & 3.15e+03 & \centering 13 & 2.24e+37 & 1.82 \\
\hline
\centering $Z_{\mathcal{U}}$-$10$-$500$ & \centering Ball & \centering $50$ & 3.19e+38 & 0.07 & 1 & 3.00e+03 & \centering 50 & 1.42e+41 & 1.84 \\
\hline
\centering $Z_{\mathcal{U}}$-$10$-$1000$ & \centering Ball & \centering $100$ & 2.57e+41 & 0.04 & 1 & 2.95e+03 & \centering 180 & 1.08e+44 & 1.83 \\
\hline
\centering $Z_{\mathcal{U}}$-$10$-$1500$ & \centering Ball & \centering $150$ & 1.85e+43 & 0.05 & 1 & 2.88e+03 & \centering 320 & 7.30e+45 & 1.82 \\
\hline
\centering $Z_{\mathcal{U}}$-$15$-$300$ & \centering Ball & \centering $20$ & 3.79e+51 & 0.07 & 1 & 4.34e+03 & \centering 50 & 3.63e+56 & 2.15 \\
\hline
\centering $Z_{\mathcal{U}}$-$15$-$750$ & \centering Ball & \centering $50$ & 8.14e+57 & 0.05 & 1 & 3.39e+03 & \centering 183 & 7.84e+62 & 2.15 \\
\hline
\centering $Z_{\mathcal{U}}$-$15$-$1500$ & \centering Ball & \centering $100$ & 3.45e+62 & 0.07 & 1 & 3.22e+03 & \centering 503 & 3.20e+62 & 2.14 \\
\hline
\centering $Z_{\mathcal{U}}$-$15$-$2250$ & \centering Ball & \centering $150$ & 1.47e+65 & 0.05 & 1 & 3.14e+03 & \centering 1071 & 1.31e+70 & 2.14 \\
\hline
\centering $Z_{\mathcal{U}}$-$20$-$300$ & \centering Ball & \centering $15$ & 1.31e+67 & 0.09 & 1 & 4.53e+03 & \centering 72 & 7.91e+74 & 2.45 \\
\hline
\centering $Z_{\mathcal{U}}$-$20$-$400$ & \centering Ball & \centering $20$ & 5.58e+69 & 0.05 & 1 & 4.46e+03 & \centering 110 & 2.90e+77 & 2.43 \\
\hline
\centering $Z_{\mathcal{U}}$-$20$-$1000$ & \centering Ball & \centering $50$ & 1.12e+77 & 0.05 & 1 & 3.76e+03 & \centering 455 & 4.91e+84 & 2.41 \\
\hline
\centering $Z_{\mathcal{U}}$-$20$-$2000$ & \centering Ball & \centering $100$ & 1.93e+83 & 0.06 & 1 & 3.51e+03 & \centering 1736 & 7.56e+90 & 2.40 \\
\hline
\centering $Z_{\mathcal{U}}$-$30$-$450$ & \centering Ball & \centering $15$ &  8.15e+99 & 0.10 & 1 & 5.45e+03 & \centering 280 & 5.73e+113 & 2.89 \\
\hline
\centering $Z_{\mathcal{U}}$-$30$-$600$ & \centering Ball & \centering $20$ & 1.49e+104 & 0.06 & 1 & 5.46e+03 & \centering 503 & 1.04e+118 & 2.89 \\
\hline
$Z_{Exp}$-$60$-$78$ & \centering Hpoly & 1.3 & 1.47e+16 & 0.07 & 1 & \centering 7.26e+04 & 357 & 9.25e+54 & 4.43 \\
\hline
$Z_{Exp}$-$80$-$104$ & \centering Hpoly & 1.3 & 6.39e+08 & 0.08 & 3.0 & \centering 1.37e+05 & 1604 & 5.03e+71 & 6.11 \\
\hline
$Z_{Exp}$-$100$-$130$ & \centering Hpoly & 1.3 & 4.73e+17 & 0.05 & 3.0 & \centering 1.97e+05  & 2851 & 1.01e+101 & 6.81 \\
\hline
$Z_{Exp}$-$60$-$90$ & \centering Hpoly & 1.5 & 9.24e+19 & 0.07 & 1 & \centering 5.68e+04  & 332 & 1.04e+58 & 4.31 \\
\hline
$Z_{Exp}$-$80$-$120$ & \centering Hpoly  & 1.5 & 3.46e+21 & 0.05 & 2.7 & \centering 7.28e+04 & 901 & 9.70e+75 & 4.79 \\
\hline
$Z_{Exp}$-$100$-$150$ & \centering Hpoly  & 1.5 & 5.30e+37 & 0.09 & 3 & \centering 1.27e+05 & 2647 & 2.66e+108 & 5.09 \\
\hline
$Z_{\mathcal{N}}$-$60$-$102$ & \centering Hpoly & 1.7 & 2.06e+132 & 0.07 & 2 & \centering 4.18e+04 & 329 & 2.65e+175  & 5.23  \\
\hline
$Z_{\mathcal{N}}$-$80$-$136$ & \centering Hpoly  & 1.7 & 1.74e+176 & 0.10 & 2.8 & \centering 5.17e+04 & 857 & 2.63e+237 & 5.81 \\
\hline
$Z_{\mathcal{N}}$-$100$-$170$ & \centering Hpoly  & 1.7 & 3.74e+222 & 0.14 & 3 & \centering 8.06e+04 & 2402 & 1.98e+303 & 6.42 \\
\hline
$Z_{\mathcal{U}}$-$60$-$120$ & \centering Hpoly  & 2 & 1.85e+135 & 0.12 & 2 & \centering 3.38e+04 & 337 & 1.89e+177 & 5.01 \\
\hline
$Z_{\mathcal{U}}$-$80$-$160$ & \centering Hpoly  & 2 & 2.27e+186 & 0.15 & 2.6 & \centering 5.72+04 & 1227 & 2.56e+245 & 5.47  \\
\hline
$Z_{\mathcal{U}}$-$100$-$200$ & \centering Hpoly  & 2 & 5.86e+231 & 0.17 & 3 & \centering 6.37e+04 & 2652 & 1.85e+310 & 6.10  \\
\hline
$Z_{\mathcal{U}}$-$60$-$180$ & \centering Hpoly  & 3 & 2.95e+155 & 0.21 & 2 & 1.98e+04 & 417 &  6.61e+191 & 4.03\\
\hline
$Z_{\mathcal{U}}$-$80$-$240$ & \centering Hpoly  & 3 & 5.59e+206 & 0.19 & 2.8 & 2.97e+04 & 2497 & 9.47e+260 & 4.76  \\
\hline
\end{tabular}
\caption{\label{fig:zono} Volume estimation for Z-polytopes by \volalg. 
For each polytope \volalg\ performs 10 runs. 
Body: type of body in MMC; 
order: the ratio between the number of generators over the dimension ($n/d$); volume: the mean computed volume; 
error: relative error;
$k$: average number of bodies in MMC;
Refl: average number of reflections (or boundary oracle calls) performed by \billiard; 
time: average time in seconds; 
$\vol(P_{red})$: volume of the over-approximation;
$R$: ratio of fitness by Equation~(\ref{pcafit}); 
$\epsilon = 0.1$ in all cases.}
\end{table}

\begin{table}[ht]
\centering
\begin{tabular}{|c||c|c|c|c|c|c|c|}\hline
 V-polytope & volume & std/mean & $k$ & Refl & error & time & exact time   \\ \hline\hline
\centering cross-$50$ & \centering 4.546e-13  & 0.05 & \centering 1.0 & 10.9e+03 & 0.05 & 3 &  $--$ \\
\hline
\centering cross-$100$ & \centering 1.43e-128 & 0.05 & \centering 2.0 & 5.80e+04 & 0.08 & 112 & $--$ \\
\hline
\centering cross-$150$ & \centering 1.50e-64 & 0.07 & \centering 2.0 & 2.99e+04 & 0.08 & 24 & $--$ \\
\hline
\centering cross-$200$ & \centering 1.65e-95 & 0.10 & \centering 3.0 & 3.92e+04 & 0.02 & 53 & $--$ \\
\hline
\centering $\Delta$-$40$ & \centering 1.17e-48 & 0.15 & \centering 6.1 & 2.42e+05 & 0.05 & 230 &  0.008 \\
\hline 
\centering $\Delta$-$60$ & \centering 1.12e-82 & 0.14 & \centering 11.2 & 5.13e+05 & 0.07 & 989 &  0.02 \\
\hline  
\centering $\Delta$-$80$ & \centering 1.35e-119 & 0.21 & \centering 15.6 & 1.21e+06 & 0.03 & 4140 & 0.07 \\
\hline    
\centering cube-$10$ & \centering 1052.4 & 0.07 & \centering 1 & 1.50e+03 &  0.03 & 54 &  $--$  \\
\hline
\centering cube-$11$ & \centering 1930.2 & 0.06 & \centering 1 & 1.5e+03 & 0.06 & 155  & $--$ \\
\hline  
\centering cube-$12$ & \centering 4240.6 & 0.08 & \centering 1 & 1.81e+03 & 0.04 & 567    &  $--$ \\
\hline        
\centering cube-$13$ & \centering 7538.2 & 0.06 & \centering 1 & 1.87e+03 & 0.08 & 2937   & $--$  \\
\hline
$K$-$60$ & \centering 2.43e-62 & 0.10 & \centering 1.0 & 4.08e+04 & ?? &  178 & $--$   \\
\hline
$K$-$70$ & \centering 2.05e-77 & 0.12 & \centering 1.0  & 4.78e+04 & ?? & 298  & $--$   \\
\hline
$K$-$80$ & \centering 3.62e-93 & 0.14 & \centering 1.0  & 5.43e+04 & ?? &  466 & $--$   \\
\hline
$K$-$90$ & \centering  1.72e-109 & 0.08 & \centering 1.9  & 7.32e+04 & ?? &  756  & $--$  \\
\hline
$K$-$100$ & \centering 3.16e-126 & 0.12 & \centering 2.0 & 7.79e+04 & ?? & 1066  & $--$   \\
\hline
$E$-$3$-$4$ & \centering 1.35e-01 & 0.07 & \centering 1.0 & 1.22e+04 & 0.02 &   35.4   & $--$  \\
\hline
$E$-$3$-$5$ & \centering  8.98e-03  & 0.05 & \centering 1.0  & 1.49e+04 & 0.00 & 358   & $--$  \\
\hline
$E$-$5$-$3$ & \centering 1.04e-01 & 0.08 & \centering 1.0 & 1.63e+04 & 0.00 & 403  &  $--$  \\
\hline
$E$-$4$-$4$ & \centering 1.43e-02 & 0.07 & \centering 1.0 & 1.74e+04 & 0.02 & 321 &  $--$  \\
\hline
$ccp$-$5$ & \centering 2.31 & 0.07 & \centering 1.0 & 9.24e+03 & 0.00 & 3.1 & 0.02  \\
\hline
$ccp$-$6$ & \centering 1.40 & 0.06 & \centering 1.0 & 1.28e+04 & 0.04 & 8.8 &  44.6  \\
\hline
$ccp$-$7$ & \centering 5.21e-01 & 0.06 & \centering 1.0 &  1.67e+04 & ?? & 29.5  & $--$  \\
\hline
$cc$-$8$-$5$ & \centering 1.77e-03 & 0.04 & \centering 1.0 & 8.79e+03  & 0.02 & 6.9 &  0.03  \\
\hline
$cc$-$8$-$7$ & \centering 2.75e+01 & 0.06 & \centering 1.0 & 7.60e+03  & 0.00 & 7.2 & 0.07  \\
\hline
$cc$-$8$-$9$ & \centering 1.36e+04 & 0.05 & \centering 1.0 & 7.38e+03  & 0.02 & 9.3 & 0.5 \\
\hline
$cc$-$8$-$11$ & \centering 1.40e+06 & 0.07 & \centering 1.0 & 7.22e+03 & 0.01 & 14.1 & 2.4 \\
\hline
\centering rvs-$10$-$20$ & \centering 1.51e-05 & 0.07 & \centering 1 & 1.11e+04 & 0.06 & 3.2 & 0.01 \\
\hline 
\centering rvc-$10$-$30$ & \centering 5.30e-02 & 0.05 & \centering 1 & 9.65e+03 & 0.05 & 4.5 & 0.09 \\
\hline 
\centering rvs-$10$-$40$ & \centering 3.85e-04 & 0.09 & \centering 1 & 7.19e+03 & 0.05 & 4.8 & 0.4 \\
\hline
\centering rvc-$10$-$50$ & \centering 0.440 & 0.05 & \centering 1 & 7.27e+03 & 0.02 & 3.7 & 0.9 \\
\hline
\centering rvs-$15$-$45$ & \centering 1.01e-03 & 0.07 & \centering 1 & 9.34e+03 & 0.04 & 6.2 & 87 \\
\hline 
\centering rvc-$15$-$60$ & \centering 8.72e-04 & 0.05 & \centering 1 & 8.98e+03 & 0.02 & 9.1 & 549 \\
\hline
\centering rvs-$20$-$40$ & \centering 2.86e-10 & 0.06 & \centering 1 & 1.23e+04 & ?? & 7.7 & $--$ \\
\hline
\centering rvc-$50$-$200$ & \centering 1.04e-18 & 0.08 & \centering 1 & 2.25e+03 & ?? & 215 & $--$ \\
\hline
\centering rvc-$60$-$90$ & \centering 1.48e-39 & 0.10 & \centering 3.1 & 5.18e+04 & ?? & 267 & $--$ \\
\hline 
\centering rvc-$80$-$120$ & \centering 5.94e-57 & 0.08 & \centering 4.5 & 7.88e+04 & ?? & 763 & $--$ \\
\hline 
\centering rvc-$100$-$150$ & \centering 1.31e-75 & 0.07 & \centering 5.8 & 1.00e+05 & ?? & 2237 & $--$ \\
\hline 
\centering rvs-$60$-$120$ & \centering 2.04e-72 & 0.11 & \centering 2.2 & 5.39e+04 & ?? & 422 & $--$ \\
\hline
\centering rvs-$80$-$160$ & \centering 5.17e-106 & 0.19 & \centering 3.0 & 7.50e+04 & ?? & 1069 & $--$ \\
\hline
\centering rvs-$100$-$200$ & \centering 9.13e-142 & 0.13 & \centering 4.1 & 1.05e+05 & ?? & 2520 & $--$ \\
\hline
\centering rvc-$60$-$120$ & \centering 2.72e-33 & 0.14 & \centering 2.1 & 5.25e+04 & ?? & 354 & $--$ \\
\hline
\centering rvc-$80$-$160$ & \centering 5.52e-49 & 0.15 & \centering 3.0 & 7.93e+04 & ?? & 1052 & $--$ \\
\hline
\centering rvc-$100$-$200$ & \centering 1.04e-65 & 0.16 & \centering 4.0 & 1.04e+05 & ?? & 2278 & $--$ \\
\hline
\end{tabular}
\caption{\label{fig:vpoly} Volume estimation for V-polytopes by \volalg. 
For each polytope \volalg\ performs 10 runs; 
volume: average estimated volume; 
$k$: average number of bodies in MMC; 
Refl: average number of reflection (or boundary oracle calls) performed by \billiard; 
error: average relative error; 
time: average time in seconds; 
exact time: time in seconds for exact volume computation by {\tt qhull} (using R package {\tt geometry}); 
$--$ implies the execution failed due to memory issues or exceeded $1$hr; 
??: the exact volume is unknown;
$\epsilon = 0.1$ in all cases.}
\end{table}

\begin{table}[ht]
\begin{tabular}{|c||c|c|c||c|c|c||c|c|c|}\hline
& \multicolumn{3}{c||}{no rounding ($C$ is a ball)} &
\multicolumn{3}{c||}{rounding} &
\multicolumn{3}{c|}{H-polytope approx.}\\ \hline
 \centering Z-polytope & $k$ & Refl & time &  $k$ &  Refl &   time &  $k$ &  Refl &  time \\ \hline\hline
$Z_{\mathcal{U}}$-$30$-$60$ &  3 & 3.09e+04 & 54 &  1 & 2.80e+04 & 50 &  \textbf{1} & \textbf{1.57e+04} & \textbf{35.9} \\ \hline
$Z_{\mathcal{U}}$-$40$-$80$ & 4 & 4.23e+04 & 126 &  1 & 3.89e+04 & 115 &  \textbf{1} & \textbf{1.71e+04} & \textbf{67.1} \\ \hline
$Z_{\mathcal{U}}$-$50$-$100$ &  5 & 5.53e+04 & 282 &  1 & 5.10e+04 & 270 &  \textbf{1} & \textbf{1.84e+04} & \textbf{133} \\ \hline
$Z_{\mathcal{U}}$-$60$-$120$ &  7 & 9.04e+04 & 825 &  1 & 6.61e+04 & 575 &  \textbf{2} & \textbf{3.52e+04} & \textbf{369} \\ \hline\hline

$Z_{\mathcal{U}}$-$30$-$150$ &  \textbf{1} & \textbf{7.01e+03} & \textbf{57}  &  1 & 1.61e+04 & 111 &  1 & 7.24e+03 & 64 \\ \hline
$Z_{\mathcal{U}}$-$40$-$200$ & \textbf{1} & \textbf{7.79e+03} & \textbf{126} &  1 & 2.07e+04 & 323 &  1 & 7.81e+03 & 163 \\ \hline
$Z_{\mathcal{U}}$-$50$-$250$ &  \textbf{1} & \textbf{8.33e+03} & \textbf{319} &  1 & 2.67e+04 & 858 &  2 & 1.24e+04 & 414 \\ \hline
$Z_{\mathcal{U}}$-$60$-$300$ &  \textbf{1} & \textbf{9.53e+03} & \textbf{721} &  1 & 3.35e+04 & 2121 &  2 & 1.37e+04 & 1168 \\ \hline
\end{tabular}
\caption{Comparisons between rounding and the H-polytope approximation  used as body in MMC of Section~\ref{sec:rounding}. 
For each Z-polytope \volalg\ performs 10 runs. 
$k$: average number of bodies in MMC; 
Refl: average number of reflection (or boundary oracle calls) performed by \billiard;  
time: average runtime of \volalg\ in seconds. 
We set $\epsilon = 0.1$ in all cases. 
Bold marks best runtimes. 
The volumes for each Z-polytope agree up to at most $\epsilon = 0.05$ and thus omitted. \label{tab:r_nr_hp2} }
\end{table}

\begin{figure}[ht!]
\centering
\includegraphics[width=\textwidth]{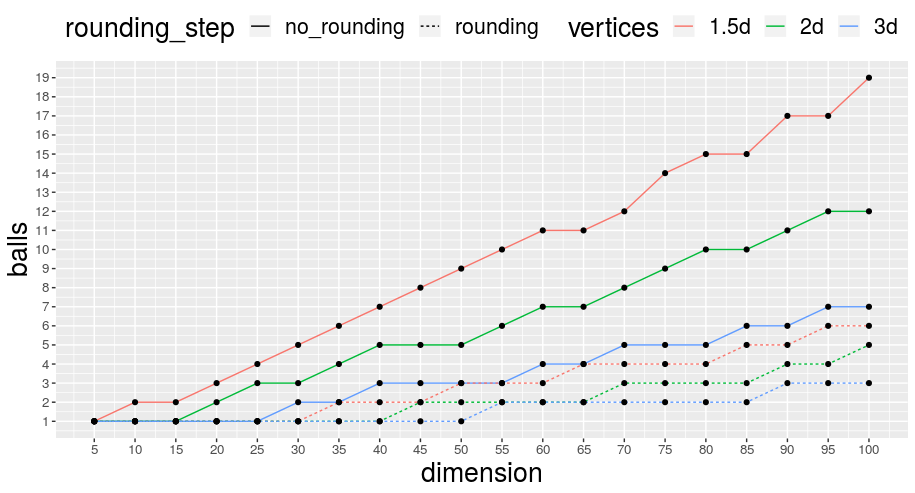}
\caption{The number of balls in MMC for $rvc$. For each dimension we generate $10$ $rvc$ and we compute the number of balls in MMC. We keep the V-polytope with the largest number of balls in MMC and then, for that one, we apply the rounding step and we compute the number of balls in MMC for the rounded polytope. \label{fig:vpoly_rounding}}
\end{figure}

\clearpage
\normalsize 

\section{Software tutorial}\label{apx:tutorial}

This section aims in providing a brief outline on how to reproduce the computational results of this article. 
First, to run the {\tt C++} implementation of \volalg\ we suggest to use the {\tt R} interface. 
To install the {\tt R} interface follow the instructions given at the repository:
\begin{center}{
\url{https://github.com/GeomScale/volume_approximation/blob/v1.1.3/doc/r_interface.md} }
\end{center}

After installation the user should open an {\tt R} terminal or alternatively use Rstudio and load the volesti package running:

\begin{lstlisting}[language=R]
library(volesti)
\end{lstlisting}

Then we start using the package by illustrating polytope generators that could be used to generate polytopes in various representations. 

\begin{lstlisting}[language=R]
d = 50
m = 200

HP = gen_cube(d, 'H')
VP = gen_cube(d, 'V')

HP = gen_simplex(d, 'H')
VP = gen_simplex(d, 'V')

HP = gen_cross(d/5, 'H')
VP = gen_cross(d/5, 'V')

# product of simplices and Birkhoff polytopes 
# are available only as H-polytopes
HP = gen_prod_simplex(d)

# 10 is the order of the Birkoff polytope
HP = gen_birkhoff(10) 

# m is the number of facets
HP = gen_rand_hpoly(d, m)  
# m is the number of vertices
VP = gen_rand_vpoly(d, m)   
# m is the number of generators
ZP = gen_rand_zonotope(d, m, generator = 'uniform')    
\end{lstlisting}
Note that when the value of the input flag of {\tt generator} in the {\tt gen\_rand\_zonotope()} is {\tt 'uniform'} then the function generates a random $Z_{\mathcal{U}}$-$d$-$n$. 
Other options are {\tt 'gaussian'} for random $Z_{\mathcal{N}}$-$d$-$n$ and {\tt 'exponential'} for random $Z_{Exp}$-$d$-$n$.

\paragraph{}
To estimate the volume of a polytope with our implementation of  \volalg\ the user should use the function {\tt volume()}. 
The output contains both the estimated volume and its logarithm (useful in overflow or underflow cases).

\begin{lstlisting}[language=R]
pair_vol = volume(HP)
# get the logarithm of the estimation
log_vol = pair_vol$log_volume
# get the volume estimation   
vol = pair_vol$volume     
\end{lstlisting}

To request polytope rounding before volume estimation using the methods we introduced in Section~\ref{sec:rounding}  run:

\begin{lstlisting}[language=R]
pair_vol = volume(HP, rounding = "isotropy")
\end{lstlisting}

To estimate the volume with algorithm \cg\ as described in~\cite{Cousins15},

\begin{lstlisting}[language=R]
pair_vol = volume(HP, setting = list("algorithm" = "CG"))
\end{lstlisting}

To compute the PCA over-approximation of a Z-polytope,

\begin{lstlisting}[language=R]
res_list = zonotope_approximation(ZP, fit_ratio = TRUE)
# The over-approximation computed with PCA
ZPred = res_list$P    
# the ratio of fitness of the computed over-approximation
fit_ratio = res_list$fit_ratio     
\end{lstlisting}

\end{document}